\documentclass[submission,copyright,creativecommons]{eptcs}
\providecommand{\event}{MSFP 2022} 

\usepackage[utf8]{inputenc}
\usepackage{graphicx}
\usepackage{fancyvrb}
\usepackage{mathtools}
\usepackage{amsmath}
\usepackage{amsthm}
\usepackage[all]{xy}

\makeatletter
\@ifundefined{lhs2tex.lhs2tex.sty.read}%
  {\@namedef{lhs2tex.lhs2tex.sty.read}{}%
   \newcommand\SkipToFmtEnd{}%
   \newcommand\EndFmtInput{}%
   \long\def\SkipToFmtEnd#1\EndFmtInput{}%
  }\SkipToFmtEnd

\newcommand\ReadOnlyOnce[1]{\@ifundefined{#1}{\@namedef{#1}{}}\SkipToFmtEnd}
\usepackage{amstext}
\usepackage{amssymb}
\usepackage{stmaryrd}
\DeclareFontFamily{OT1}{cmtex}{}
\DeclareFontShape{OT1}{cmtex}{m}{n}
  {<5><6><7><8>cmtex8
   <9>cmtex9
   <10><10.95><12><14.4><17.28><20.74><24.88>cmtex10}{}
\DeclareFontShape{OT1}{cmtex}{m}{it}
  {<-> ssub * cmtt/m/it}{}

\DeclareFontShape{OT1}{cmtt}{bx}{n}
  {<5><6><7><8>cmtt8
   <9>cmbtt9
   <10><10.95><12><14.4><17.28><20.74><24.88>cmbtt10}{}
\DeclareFontShape{OT1}{cmtex}{bx}{n}
  {<-> ssub * cmtt/bx/n}{}

\newcommand{\Conid}[1]{\mathit{#1}}
\newcommand{\Varid}[1]{\mathit{#1}}
\newcommand{\anonymous}{\kern0.06em \vbox{\hrule\@width.5em}}
\newcommand{\plus}{\mathbin{+\!\!\!+}}
\newcommand{\bind}{\mathbin{>\!\!\!>\mkern-6.7mu=}}
\newcommand{\sequ}{\mathbin{>\!\!\!>}}

\renewcommand{\geq}{\geqslant}
\usepackage{polytable}

\@ifundefined{mathindent}%
  {\newdimen\mathindent\mathindent\leftmargini}%
  {}%

\def\resethooks{%
  \global\let\SaveRestoreHook\empty
  \global\let\ColumnHook\empty}
\newcommand*{\savecolumns}[1][default]%
  {\g@addto@macro\SaveRestoreHook{\savecolumns[#1]}}
\newcommand*{\restorecolumns}[1][default]%
  {\g@addto@macro\SaveRestoreHook{\restorecolumns[#1]}}
\newcommand*{\aligncolumn}[2]%
  {\g@addto@macro\ColumnHook{\column{#1}{#2}}}

\resethooks

\newcommand{\onelinecommentchars}{\quad-{}- }
\newcommand{\commentbeginchars}{\enskip\{-}
\newcommand{\commentendchars}{-\}\enskip}

\newcommand{\visiblecomments}{%
  \let\onelinecomment=\onelinecommentchars
  \let\commentbegin=\commentbeginchars
  \let\commentend=\commentendchars}

\newcommand{\invisiblecomments}{%
  \let\onelinecomment=\empty
  \let\commentbegin=\empty
  \let\commentend=\empty}

\visiblecomments

\newlength{\blanklineskip}
\newcommand{\hsindent}[1]{\quad}
\let\hspre\empty
\let\hspost\empty

\EndFmtInput
\makeatother
\ReadOnlyOnce{polycode.fmt}%
\makeatletter

\newcommand{\hsnewpar}[1]%
  {{\parskip=0pt\parindent=0pt\par\vskip #1\noindent}}

\newcommand{\hscodestyle}{}

\newcommand{\sethscode}[1]%
  {\expandafter\let\expandafter\hscode\csname #1\endcsname
   \expandafter\let\expandafter\endhscode\csname end#1\endcsname}

  {\par\noindent
   \advance\leftskip\mathindent
   \hscodestyle
   \let\\=\@normalcr
   \let\hspre\(\let\hspost\)%
   \pboxed}%
  {\endpboxed\)%
   \par\noindent
   \ignorespacesafterend}

  {\hsnewpar\abovedisplayskip
   \advance\leftskip\mathindent
   \hscodestyle
   \let\hspre\(\let\hspost\)%
   \pboxed}%
  {\endpboxed%
   \hsnewpar\belowdisplayskip
   \ignorespacesafterend}

  {\hsnewpar\abovedisplayskip
   \advance\leftskip\mathindent
   \hscodestyle
   \let\\=\@normalcr
   \(\pboxed}%
  {\endpboxed\)%
   \hsnewpar\belowdisplayskip
   \ignorespacesafterend}

\newcommand{\plainhs}{\sethscode{plainhscode}}

\plainhs

  {\hsnewpar\abovedisplayskip
   \advance\leftskip\mathindent
   \hscodestyle
   \let\\=\@normalcr
   \(\parray}%
  {\endparray\)%
   \hsnewpar\belowdisplayskip
   \ignorespacesafterend}

  {\parray}{\endparray}

  {\(\parray}{\endparray\)}

\def\codeframewidth{\arrayrulewidth}
\RequirePackage{calc}

  {\parskip=\abovedisplayskip\par\noindent
   \hscodestyle
   \arrayrulewidth=\codeframewidth
   \tabular{@{}|p{\linewidth-2\arraycolsep-2\arrayrulewidth-2pt}|@{}}%
   \hline\framedhslinecorrect\\{-1.5ex}%
   \let\endoflinesave=\\
   \let\\=\@normalcr
   \(\pboxed}%
  {\endpboxed\)%
   \framedhslinecorrect\endoflinesave{.5ex}\hline
   \endtabular
   \parskip=\belowdisplayskip\par\noindent
   \ignorespacesafterend}

\newcommand{\framedhslinecorrect}[2]%
  {#1[#2]}

  {\(\def\column##1##2{}%
   \let\>\undefined\let\<\undefined\let\\\undefined
   \newcommand\>[1][]{}\newcommand\<[1][]{}\newcommand\\[1][]{}%
   \def\fromto##1##2##3{##3}%
   }{\) }%

  {\let\orighscode=\hscode
   \let\origendhscode=\endhscode
   \def\endhscode{\def\hscode{\endgroup\def\@currenvir{hscode}\\}\begingroup}
   \orighscode\def\hscode{\endgroup\def\@currenvir{hscode}}}%
  {\origendhscode
   \global\let\hscode=\orighscode
   \global\let\endhscode=\origendhscode}%

\makeatother
\EndFmtInput
\ReadOnlyOnce{forall.fmt}%
\makeatletter

\let\HaskellResetHook\empty
\newcommand*{\AtHaskellReset}[1]{%
  \g@addto@macro\HaskellResetHook{#1}}
\newcommand*{\HaskellReset}{\HaskellResetHook}

\newcommand\hsforall{\global\let\hsdot=\hsperiodonce}
\newcommand*\hsperiodonce[2]{#2\global\let\hsdot=\hscompose}
\newcommand*\hscompose[2]{#1}

\AtHaskellReset{\global\let\hsdot=\hscompose}

\HaskellReset

\makeatother
\EndFmtInput

\newtheorem{theorem}{Theorem}[section]
\newtheorem{prop}[theorem]{Proposition}

\newtheorem*{proof*}{Proof}
\newtheorem*{sk*}{Sketch}
\newtheorem{definition}[theorem]{Definition}

\newcommand{\CC}{\mathcal{C}}
\newcommand{\DD}{\mathcal{D}}
\newlength{\nl}
\title{On Structuring Functional Programs with Monoidal Profunctors}
\author{Alexandre Garcia de Oliveira
\institute{IME-USP\\ São Paulo, Brazil}
\institute{Instituto de Matemática e Estatística\\
Univerisidade de São Paulo\\
São Paulo, Brazil}
\email{alexgrcol (at) hotmail.com}
\and
Mauro Jaskelioff
\institute{CIFASIS-CONICET\\
Rosario, Argentina}
\institute{FCEIA\\
    Universidad Nacional de Rosario\\
    Argentina}
\email{jaskelioff(at)cifasis-conicet.gov.ar}
\and
Ana Cristina Vieira de Melo
\institute{IME-USP\\
São Paulo, Brazil}
\email{acvm(at)ime.usp.br}
}
\def\titlerunning{On Structuring Functional Programs with Monoidal Profunctors}
\def\authorrunning{A.G. Oliveira, M. Jaskelioff \& A.C. V. de Melo}

\begin{document}

\maketitle

\begin{abstract}
 We study monoidal profunctors as a tool to reason and structure pure functional programs both from a categorical perspective and as a Haskell implementation. From the categorical point of view we approach them as monoids in a certain monoidal category of profunctors. We study properties of this monoidal category and construct and implement the free monoidal profunctor. We study the relationship of the monoidal construction to optics, and introduce a promising generalization of the implementation which we illustrate by introducing effectful monoidal profunctors. 
\end{abstract}

\section{Introduction}
Functors, applicative functors~\cite{Mcbride}, monads~\cite{Moggi,Wadler92a}, profunctors, and arrows~\cite{Hughes2004, Hasuo} are by now part of the vocabulary of 
the programmer who writes mathematically structured programs. In order to understand and develop these structures both the categorical view 
and the programming view have been helpful. For example, Lindley et al.~\cite{Lindwad} compare these structures from the point of view of (typed) programming languages, 
and~Rivas and Jaskelioff~\cite{Jask2017} compare them from the point of view of monoidal categories. 
In this last work, both monads and applicative functors are seen as monoids in a monoidal category of endofunctors. 
Monads use functor composition as tensor, whereas applicative functors use the Day convolution as tensor.
Likewise, arrows can be seen as monoids in a monoidal category of (strong) profunctors.

However, given that the Day convolution can also be a tensor of profunctors there is another structure yet to be studied.
In this work we study monoidal profunctors, a structure obtained from considering monoids in a category of profunctors 
taking the Day convolution as tensor, hence filling the gap in the following table:

\begin{table}[h!]
  \begin{center}
    \begin{tabular}{|c|c|c|}
      \hline
      functor & applicative & monad\\
      \hline
      profunctor & ???? & arrow \\
      \hline
      \end{tabular}
      \medskip
      \caption{Structure relations}
      \label{tab:table1}
  \end{center}
\end{table}






Monoidal profunctors are a categorical structure with two components: an
identity computation and a generic parallel composition.

Therefore, with this paper we aim to gather the knowledge about monoidal
profunctors and study their application in the context of functional programming.

After some mathematical background (Section~\ref{sec:mathematicalbackground}), we introduce monoidal profunctors as monoids in a monoidal category 
and study some of their properties (Section~\ref{sec:monoidalprofunctors}). Then, in Section~\ref{sec:haskellimplementation}, we implement these ideas in Haskell and provide some instances.
We show the free monoidal profunctor and a Haskell implementation in Section~\ref{sec:freemonoidalprofunctor}.
In Section~\ref{sec:monocles}, we present an application of monoidal profunctors to optics~\cite{bartosz}
and observe connections with well-known structures such as traversals and grates~\cite{Wu, connor}. 
In Section~\ref{sec:effectfulmonoidprofunctor}, we generalize the implementation to other categories and illustrate its application to effectful monoidal profunctors.


\section{Mathematical background}
\label{sec:mathematicalbackground}
\subsection{Monoidal Categories and Monoids}
A \emph{monoidal category} \cite{Jask2017} gives us a minimal framework for defining the categorical version of a monoid.  

\begin{definition}
\label{monocat}
  A \emph{monoidal category} is a sextuple ($\mathcal{C}$, $\otimes$, $I$,
  $\alpha$, $\rho$, $\lambda$) where 
  \begin{itemize}
    \item $\mathcal{C}$ is a category;
    \item $\otimes:\mathcal{C} \times \mathcal{C} \rightarrow \mathcal{C}$ is
  a bifunctor;
    \item $I$ is an object called unit;
    \item $\rho_A : A \otimes I\rightarrow A$, $\lambda_A : I \otimes A \rightarrow A$ and $\alpha_{ABC} : (A \otimes B) \otimes C \rightarrow A \otimes (B \otimes C)$ are three natural isomorphisms.
  \end{itemize}
\end{definition}

If the isomorphisms $\rho$, $\lambda$ and $\alpha$ are identities then the monoidal category is called \emph{strict}, if there is a natural isomorphism
$\gamma_{AB} : A \otimes B \rightarrow B \otimes A$ the monoidal category is called \emph{symmetric}.

A monoidal category is \emph{closed} if there is an additional functor, called the internal hom, 
$\Rightarrow : \mathcal{C}^{op} \times \mathcal{C} \rightarrow Set$ such that 
$\CC (A \otimes B, C) \cong \CC(A,B \Rightarrow C)$, natural in $A$, $B$ and $C$, objects of $\mathcal{C}$. 
The witnesses of this isomorphism are called currying and uncurrying. In $Set$, with $\otimes = \times$, $A \Rightarrow B$ is just the hom-set $A \rightarrow B$.


\begin{definition}
A \emph{monoid} in a monoidal category $\mathcal{C}$ is the tuple $(M,e,m)$ where 
$M$ is an object of $\mathcal{C}$, $e:I \rightarrow M$ is the unit morphism
and $m: M \otimes M \rightarrow M$ is the multiplication morphism, satisfying
\begin{enumerate}
\item Right unit: $m \circ (id \otimes e) = \rho_{MMM}$
\item Left unit: $m \circ (e \otimes id) = \lambda_{MMM}$
\item Associativity: $m \circ (m \otimes id) = m \circ (id \otimes m) \circ \alpha_{MMM} $
\end{enumerate}
\end{definition}

\subsection{Profunctors}

A \emph{profunctor} generalizes the notions of function relation and bimodule \cite{operads}.

\begin{definition}
Given two categories $\mathcal{C}$ and $\mathcal{D}$, a profunctor from $\CC$ to $\DD$ is a functor $P : \mathcal{C}^{op} \times \mathcal{D} \rightarrow Set$, consisting of:
\begin{itemize}
    \item for each $a$ object of $\mathcal{C}$ and $b$ object of $\mathcal{D}$, a set $P(a,b)$;
    \item for each $a$ object of $\mathcal{C}$ and $b, d$ objects of $\mathcal{D}$, a function (left action) $\mathcal{D}(d,b) \times P(a,d) \rightarrow P(a,b)$;
    \item for all $a,c$ objects of $\mathcal{C}$ and $b$ object of $\mathcal{D}$, a function (right action) $P(a,b) \times \CC(c,a) \rightarrow P(c,b)$.
\end{itemize}
\end{definition}

This notion is also known as a Bimodule or a ($\mathcal{C}$,$\mathcal{D}$)-module, and also as a distributor. 

Since a profunctor is a functor from the product category $\CC^{op} \times \DD$ to $Set$, it must satisfy the functor laws.%
\begin{align*}
P(1_{C},1_{D}) &= 1_{P(C,D)} \\
P(f\circ g,h \circ i) &= P(g,h) \circ P(f,i) 
\end{align*}

An example of a profunctor is the hom functor $Hom : \mathcal{C}^{op} \times \mathcal{C} \rightarrow Set$, written as $A \rightarrow B$ when $\CC = Set$, and its actions are just pre-composition and post-composition of set functions.

\begin{definition}
Let $\CC$ and $\DD$ be small categories, $Prof(\CC,\DD)$ is the profunctor category consisting of profunctors $\CC^{op} \times \DD \to Set$ as objects, natural transformations between profunctors as morphisms, and vertical composition to compose them. The fact that $\CC$ and $\DD$ are small categories will always be implied when some statement about the category $Prof(\CC,\DD)$ is present. We will call this category $Prof$ when the categories $\CC$ and $\DD$ are clear from the context.
\end{definition}

The profunctor category inherits some structure from the category $Set$ such as binary products given by $(P \times Q) (S,T) = P(S,T) \times Q(S,T)$ and binary coproducts given by $(P + Q) (S,T) = P(S,T) + Q(S,T)$, where $\times, +$ are the respective universal constructions from $Set$. There are also terminal and initial profunctors given by $1_p(S,T) = \{*\}$ and $0_p(S,T) = \emptyset$, i.e., they are just constant maps to the initial and terminal object of $Set$, respectively.
\subsection{Day Convolution}

\begin{definition}
Let $\mathcal{E}$ be a small monoidal category and $F,G : \mathcal{E} \rightarrow Set$, then the \emph{Day convolution} \cite{Day70} of F and G is another functor (in $T$) given by the coend
\begin{equation}
    (F \star G) T = \int^{XY} FX \times GY \times \mathcal{D}(X \otimes Y, T).
\end{equation}
\end{definition}

One can work with this coend in the category $Prof(\CC,\DD)$ of profunctors, that is, letting $\mathcal{E} = \mathcal{C}^{op} \times \mathcal{D}$ in the above definition, enabling us to derive the following Day convolution for profunctors.
\[
     (P \star Q)(S,T) \cong \int^{ABCD} P(A,B) \times Q(C,D) \times \mathcal{C}(S, A \otimes C) \times \mathcal{C}(B \otimes D, T) 
\]

\noindent The profunctor $J(A,B) = \CC(A,I) \times \CC(I,B)$ is a unit for $\star$. When $I = 1$, where $1$ is the terminal object, then $J(A,B) \cong B$. When $\CC = Set$, this convolution gives rise to a monoidal profunctor (see next section), when $\CC$ is a Kleisli category on Set (same objects, but morphisms relies on a monad) gives the notion of an effectful monoidal profunctor (see Section~\ref{sec:effectfulmonoidprofunctor}).

\begin{prop}
\label{unit}
Let $\CC$ be a monoidal category, the profunctor $J(A,B) = \CC(A,I) \times \CC(I,B)$ is the right and left unit of $\star$.
\begin{proof}
The calculation is standard coend calculus \cite{coend} using Yoneda's lemma. This is a proof for $J$ being a right unit, for the left one is analogous.
\begin{flalign*}
    (P \star J)(S,T) & = \int^{ABCD} P(A,B) \times J(C,D) \times \CC(S, A \otimes C) \times \CC(B \otimes D, T) \\
    & \cong \int^{ABCD} P(A,B) \times \CC(C,I) \times \CC(I,D) \times \CC(S, A \otimes C) \times \CC(B \otimes D, T) \\
    & \cong \int^{ABD} P(A,B) \times \CC(I,D) \times \CC(S, A \otimes I) \times \CC(B \otimes D, T) \\
    & \cong \int^{AB} P(A,B) \times \CC(S, A \otimes I) \times \CC(B \otimes I, T) \\
    & \cong \int^{AB} P(A,B) \times \CC(S, A) \times \CC(B, T) \\
    & \cong P(S,T)
\end{flalign*}
\end{proof}
\end{prop}

The associativity of $\star$ is required to define a monoidal profunctor category.

\begin{prop}
\label{assoc}
Let ($\CC$,$\otimes$,$I$) be a monoidal category and $S,T$ two objects of $\CC$, the Day convolution for profunctors is an associative tensor product $(P \star Q) \star R \cong P \star (Q \star R)$.
\begin{proof}
The proof follows the same coend calculus pattern using Yoneda's lemma whenever needed.
\end{proof}
\end{prop}

In order to be able to define monoids in a monoidal profunctor category, one needs to check that when $\CC$ and $\DD$ are monoidal categories then $(Prof(\CC,\DD)), \star, J)$ is a monoidal category. 

\begin{theorem}\label{thm:monoidalprofunctorcategory}
Let $\CC$ and $\DD$ are monoidal small categories. Then $(Prof(\CC,\DD)), \star, J)$ is a monoidal category.
\begin{proof} 
Since $\CC$ and $\DD$ are monoidal categories, $\star$ is a bifunctor by construction, and by Proposition \ref{unit} and \ref{assoc} gives the desired morphisms, it follows that $(Prof(\CC,\DD)), \star, J)$ is a monoidal category.   
\end{proof}
\end{theorem}

Having obtained a monoidal category of profunctors, it is now possible to define a monoid in this category. In order to do that, we will use the following proposition (as in the work of Rivas and Jaskelioff \cite{Jask2017}).  

\begin{prop}
\label{monprof}
Let $\DD = \CC^{op} \otimes \CC$, there is a one-to-one correspondence defining morphisms going out of a Day convolution for profunctors
\begin{equation*}
\int_{XY} (P\star Q)(X,Y) \to R(X,Y) \cong \int_{ABCD} P(A,B) \times Q(C,D) \to R(A \otimes C,B \otimes D)
\end{equation*}
which is natural in $P$, $Q$ and $R$.
\begin{proof*}
This proof uses the same coend calculus pattern with the help of Yoneda lemma and the fact that the hom functor commutes with ends and coends \cite{coend}.
\begin{align*}
&\int_{XY} (P\star Q)(X,Y) \to R(X,Y)  \\
\cong & \int_{XY} (\int^{ABCD} (P(A,B) \times Q(C,D)) \times \CC(X, A \otimes C) \times \CC(B \otimes D,Y)) \to R(X,Y) \\
\cong& \int_{XYABCD} (P(A,B) \times Q(C,D)) \times \CC(X, A \otimes C) \times \CC(B \otimes D,Y) \to R(X,Y) \\
\cong& \int_{XYABCD} (P(A,B) \times Q(C,D)) \to \CC(X, A \otimes C) \to \CC(B \otimes D,Y) \to R(X,Y) \\
\cong& \int_{YABCD} (P(A,B) \times Q(C,D)) \to \CC(B \otimes D,Y) \to R(A \otimes C,Y) \\ 
\cong& \int_{ABCD} P(A,B) \times Q(C,D) \to R(A \otimes C,B \otimes D)
\end{align*}

\end{proof*}
\end{prop}

Whenever $P = Q = R$ in the equation of Proposition $\ref{monprof}$ we get the following isomorphism, useful to define a monoid in the profunctor category $Prof$ with Day convolution as its tensor.
\[
    \int_{XY} (P\star P)(X,Y) \to P(X,Y) \cong \int_{ABCD} P(A,B) \times P(C,D)\to P(A \otimes C,B \otimes D)
\]

\section{Monoidal Profunctors}
\label{sec:monoidalprofunctors}

We define \emph{monoidal profunctors} as monoids in the monoidal category of profunctors of theorem~\ref{thm:monoidalprofunctorcategory}.

\begin{prop}
\label{unitmonopro}
Let ($\CC$,$\otimes$,$I$) be a small monoidal category, $P: \CC^{op} \times \CC \to Set$ be a profunctor, and $S,T$ two objects of $\CC$. Then $\CC(J(S,T),P(S,T)) \cong P(I,I)$.
\end{prop}
\begin{proof}
\begin{align*}
\CC(J(S,T),P(S,T)) &\cong \CC(S,I) \times \CC(I,T) \to P(S,T) \\
& \cong \CC(S,I) \to \CC(I,T) \to P(S,T) \\
& \cong \CC(S,I) \to P(S,I) \\
& \cong P(I,I)
\end{align*}
\end{proof}

With all categorical tools in hand, the central notion of this work emerges from the category of monoidal profunctors.

\begin{definition}
Let $(\CC,\otimes,I)$ and $(\DD,\otimes,I)$ be small monoidal categories. A monoid in the monoidal profunctor category $Prof(\CC,\DD)$ consists of a profunctor $P$, a unit $e : P(I,I)$, and a multiplication given by a natural family of morphisms $m_{ABCD} : P(A,B) \times P(C,D)\to P(A \otimes C,B \otimes D)$. 
\end{definition}

\noindent
The unit is a natural transformation $e : J \to P$, which by proposition~\ref{unitmonopro} is isomorphic to $e : P(I,I)$. The multiplication is a natural transformation $m : P \star P \to P$, which by~proposition~\ref{monprof} is equivalent to the family above.

As an example, consider $(Set,\otimes,I)$, where $I$ is a singleton set, and the $Hom$ profunctor $P(A,B) = A \to B$, trivially gives us a monoidal profunctor.


\begin{prop}
\label{exp}
Let $(\CC,\otimes,I)$ and $(\DD,\otimes,J)$ be a small monoidal categories, and $P, Q$ monoidal profunctors, then 
\[
(P \Rightarrow Q)(X,Y) = \int_{CD} P(C,D) \rightarrow Q(X \otimes C,Y \otimes D) 
\]
defines an internal hom on the monoidal profunctor category $Prof(\CC,\DD)$.
\begin{proof}
This proof follows the same steps as in the functor case \cite{Jask2017} but adapting it for $Prof(\CC,\DD)$.
\end{proof}
\end{prop}
This proposition states that the monoidal category of profunctors $Prof$ is closed.

\section{Implementation in Haskell}
We implement in Haskell the concepts of profunctors, Day convolution, and monoidal profunctors defined before. Although other monoidal profunctors can be derived using other bifunctors, this work focuses only on the product. This section considers the (fictitious) category $Hask$ as a small monoidal category and $Prof(Hask,Hask)$ as the small category of profunctors to implement the typeclass \ensuremath{\Conid{MonoPro}} which represents the desired monoid over this category.
\label{sec:haskellimplementation}
\subsection{Profunctors}
\label{sec:profunctortypeclass}
A profunctor is an instance of the following class

\begin{hscode}\SaveRestoreHook
\column{B}{@{}>{\hspre}l<{\hspost}@{}}%
\column{5}{@{}>{\hspre}l<{\hspost}@{}}%
\column{E}{@{}>{\hspre}l<{\hspost}@{}}%
\>[B]{}\mathbf{class}\;\Conid{Profunctor}\;\Varid{p}\;\mathbf{where}{}\<[E]%
\\
\>[B]{}\hsindent{5}{}\<[5]%
\>[5]{}\Varid{dimap}\mathbin{::}(\Varid{a}\to \Varid{b})\to (\Varid{c}\to \Varid{d})\to \Varid{p}\;\Varid{b}\;\Varid{c}\to \Varid{p}\;\Varid{a}\;\Varid{d}{}\<[E]%
\ColumnHook
\end{hscode}\resethooks

A profunctor is a functor, thus \ensuremath{\Varid{dimap}} needs to satisfy the functor laws.


The profunctor interface lifts pure functions into both type arguments, the first in a contravariant manner, and the second in a covariant way. 
A morphism in the $Prof$ category can be represented, in Haskell, as the type below.

\begin{hscode}\SaveRestoreHook
\column{B}{@{}>{\hspre}l<{\hspost}@{}}%
\column{E}{@{}>{\hspre}l<{\hspost}@{}}%
\>[B]{}\mathbf{type}\;(\leadsto)\;\Varid{p}\;\Varid{q}\mathrel{=}\forall \Varid{x}\hsforall \;\Varid{y}\hsdot{\circ }{.}\Varid{p}\;\Varid{x}\;\Varid{y}\to \Varid{q}\;\Varid{x}\;\Varid{y}{}\<[E]%
\ColumnHook
\end{hscode}\resethooks

The function type \ensuremath{(\to )}, is the most basic example of a profunctor. 

One notion captured by a Profunctor is that of a structured input and structured output of a function \ensuremath{\Conid{SISO}}. This type generalizes Kleisli arrows which allow a pure input and a structured output. 
\begin{hscode}\SaveRestoreHook
\column{B}{@{}>{\hspre}l<{\hspost}@{}}%
\column{5}{@{}>{\hspre}l<{\hspost}@{}}%
\column{E}{@{}>{\hspre}l<{\hspost}@{}}%
\>[B]{}\mathbf{data}\;\Conid{SISO}\;\Varid{f}\;\Varid{g}\;\Varid{a}\;\Varid{b}\mathrel{=}\Conid{SISO}\;\{\mskip1.5mu \Varid{unSISO}\mathbin{::}\Varid{f}\;\Varid{a}\to \Varid{g}\;\Varid{b}\mskip1.5mu\}{}\<[E]%
\\[\blanklineskip]%
\>[B]{}\mathbf{instance}\;(\Conid{Functor}\;\Varid{f},\Conid{Functor}\;\Varid{g})\Rightarrow \Conid{Profunctor}\;(\Conid{SISO}\;\Varid{f}\;\Varid{g})\;\mathbf{where}{}\<[E]%
\\
\>[B]{}\hsindent{5}{}\<[5]%
\>[5]{}\Varid{dimap}\;\Varid{ab}\;\Varid{cd}\;(\Conid{SISO}\;\Varid{bc})\mathrel{=}\Conid{SISO}\;(\Varid{fmap}\;\Varid{cd}\hsdot{\circ }{.}\Varid{bc}\hsdot{\circ }{.}\Varid{fmap}\;\Varid{ab}){}\<[E]%
\ColumnHook
\end{hscode}\resethooks


\subsection{The Day convolution type}

The Day convolution is represented by the existential type

\begin{hscode}\SaveRestoreHook
\column{B}{@{}>{\hspre}l<{\hspost}@{}}%
\column{E}{@{}>{\hspre}l<{\hspost}@{}}%
\>[B]{}\mathbf{data}\;\Conid{Day}\;\Varid{p}\;\Varid{q}\;\Varid{s}\;\Varid{t}\mathrel{=}\forall \Varid{a}\hsforall \;\Varid{b}\;\Varid{c}\;\Varid{d}\hsdot{\circ }{.}\Conid{Day}\;(\Varid{p}\;\Varid{a}\;\Varid{b})\;(\Varid{q}\;\Varid{c}\;\Varid{d})\;(\Varid{s}\to (\Varid{a},\Varid{c}))\;(\Varid{b}\to \Varid{d}\to \Varid{t}){}\<[E]%
\ColumnHook
\end{hscode}\resethooks

Since $\CC(A,I)$ is isomorphic to a singleton set (unit of the cartesian product $\times$), and $\CC(I,B) \cong B$, one can write, in Haskell, the type

\begin{hscode}\SaveRestoreHook
\column{B}{@{}>{\hspre}l<{\hspost}@{}}%
\column{E}{@{}>{\hspre}l<{\hspost}@{}}%
\>[B]{}\mathbf{data}\;\Conid{I}\;\Varid{a}\;\Varid{b}\mathrel{=}\Conid{I}\;\{\mskip1.5mu \Varid{unI}\mathbin{::}\Varid{b}\mskip1.5mu\}{}\<[E]%
\ColumnHook
\end{hscode}\resethooks

\noindent as the unit of the Day convolution. The following functions are representations of the right and left units.

\begin{hscode}\SaveRestoreHook
\column{B}{@{}>{\hspre}l<{\hspost}@{}}%
\column{E}{@{}>{\hspre}l<{\hspost}@{}}%
\>[B]{}\rho\mathbin{::}\Conid{Profunctor}\;\Varid{p}\Rightarrow \Conid{Day}\;\Varid{p}\;\Conid{I}\leadsto\Varid{p}{}\<[E]%
\\
\>[B]{}\rho\;(\Conid{Day}\;\Varid{pab}\;(\Conid{I}\;\Varid{d})\;\Varid{sac}\;\Varid{bdt})\mathrel{=}\Varid{dimap}\;(\Varid{fst}\hsdot{\circ }{.}\Varid{sac})\;(\lambda \Varid{b}\to \Varid{bdt}\;\Varid{b}\;\Varid{d})\;\Varid{pab}{}\<[E]%
\\[\blanklineskip]%
\>[B]{}\lambda\mathbin{::}\Conid{Profunctor}\;\Varid{q}\Rightarrow \Conid{Day}\;\Conid{I}\;\Varid{q}\leadsto\Varid{q}{}\<[E]%
\\
\>[B]{}\lambda\;(\Conid{Day}\;(\Conid{I}\;\Varid{b})\;\Varid{qcd}\;\Varid{sac}\;\Varid{bdt})\mathrel{=}\Varid{dimap}\;(\Varid{snd}\hsdot{\circ }{.}\Varid{sac})\;(\lambda \Varid{d}\to \Varid{bdt}\;\Varid{b}\;\Varid{d})\;\Varid{qcd}{}\<[E]%
\ColumnHook
\end{hscode}\resethooks

The associativity of the Day convolution and its symmetric map also can be represented in Haskell as the functions below.

\begin{hscode}\SaveRestoreHook
\column{B}{@{}>{\hspre}l<{\hspost}@{}}%
\column{5}{@{}>{\hspre}l<{\hspost}@{}}%
\column{9}{@{}>{\hspre}l<{\hspost}@{}}%
\column{23}{@{}>{\hspre}l<{\hspost}@{}}%
\column{26}{@{}>{\hspre}l<{\hspost}@{}}%
\column{E}{@{}>{\hspre}l<{\hspost}@{}}%
\>[B]{}\alpha \mathbin{::}(\Conid{Profunctor}\;\Varid{p},\Conid{Profunctor}\;\Varid{q},\Conid{Profunctor}\;\Varid{r})\Rightarrow \Conid{Day}\;(\Conid{Day}\;\Varid{p}\;\Varid{q})\;\Varid{r}\leadsto\Conid{Day}\;\Varid{p}\;(\Conid{Day}\;\Varid{q}\;\Varid{r}){}\<[E]%
\\
\>[B]{}\alpha \;(\Conid{Day}\;(\Conid{Day}\;\Varid{p}\;\Varid{q}\;\Varid{s}_{1}\;\Varid{f})\;\Varid{r}\;\Varid{s}_{2}\;\Varid{g})\mathrel{=}\Conid{Day}\;\Varid{p}\;(\Conid{Day}\;\Varid{q}\;\Varid{r}\;\Varid{f}_{1}\;\Varid{f}_{2})\;\Varid{f}_{3}\;\Varid{f}_{4}{}\<[E]%
\\
\>[B]{}\hsindent{5}{}\<[5]%
\>[5]{}\mathbf{where}{}\<[E]%
\\
\>[5]{}\hsindent{4}{}\<[9]%
\>[9]{}\Varid{f}_{1}{}\<[23]%
\>[23]{}\mathrel{=}\Varid{first'}\;(\Varid{snd}\hsdot{\circ }{.}\Varid{s}_{1})\hsdot{\circ }{.}\Varid{s}_{2}{}\<[E]%
\\
\>[5]{}\hsindent{4}{}\<[9]%
\>[9]{}\Varid{f}_{2}\;\Varid{d}_{1}\;\Varid{d}_{2}{}\<[23]%
\>[23]{}\mathrel{=}{}\<[26]%
\>[26]{}(\Varid{d}_{2},\lambda \Varid{x}\to \Varid{f}\;\Varid{x}\;\Varid{d}_{1}){}\<[E]%
\\
\>[5]{}\hsindent{4}{}\<[9]%
\>[9]{}\Varid{f}_{3}{}\<[23]%
\>[23]{}\mathrel{=}\Varid{first'}\;(\Varid{fst}\hsdot{\circ }{.}\Varid{s}_{1}\hsdot{\circ }{.}(\Varid{fst}\hsdot{\circ }{.}\Varid{s}_{2}))\hsdot{\circ }{.}\Varid{diag}{}\<[E]%
\\
\>[5]{}\hsindent{4}{}\<[9]%
\>[9]{}\Varid{f}_{4}\;\Varid{b}_{1}\;(\Varid{d}_{2},\Varid{h}){}\<[23]%
\>[23]{}\mathrel{=}\Varid{g}\;(\Varid{h}\;\Varid{b}_{1})\;\Varid{d}_{2}{}\<[E]%
\\[\blanklineskip]%
\>[B]{}\gamma \mathbin{::}(\Conid{Profunctor}\;\Varid{p},\Conid{Profunctor}\;\Varid{q})\Rightarrow \Conid{Day}\;\Varid{p}\;\Varid{q}\leadsto\Conid{Day}\;\Varid{q}\;\Varid{p}{}\<[E]%
\\
\>[B]{}\gamma \;(\Conid{Day}\;\Varid{p}\;\Varid{q}\;\Varid{sac}\;\Varid{bdt})\mathrel{=}\Conid{Day}\;\Varid{q}\;\Varid{p}\;(\Varid{swap}\hsdot{\circ }{.}\Varid{sac})\;(\Varid{flip}\;\Varid{bdt}){}\<[E]%
\\
\>[B]{}\hsindent{5}{}\<[5]%
\>[5]{}\mathbf{where}\;\Varid{swap}\;(\Varid{x},\Varid{y})\mathrel{=}(\Varid{y},\Varid{x}){}\<[E]%
\ColumnHook
\end{hscode}\resethooks

Functions \ensuremath{\rho }, \ensuremath{\lambda }, and \ensuremath{\alpha } are natural isomorphisms. We leave the definition of the inverses as an exercise for the reader.

\subsection{MonoPro typeclass}
We define a typeclass called \ensuremath{\Conid{MonoPro}} for implementing monoidal profunctors.
The type \ensuremath{\Varid{p}\;()\;()} is a representation in Haskell of the unit $P(I,I)$.
The multiplication is obtained from Proposition~\ref{monprof}, which gives the multiplication a type $\int_{ABCD} P(A,B) \times P(C,D)\to P(A \otimes C,B \otimes D)$ allowing to write the following class in Haskell.

\begin{hscode}\SaveRestoreHook
\column{B}{@{}>{\hspre}l<{\hspost}@{}}%
\column{5}{@{}>{\hspre}l<{\hspost}@{}}%
\column{E}{@{}>{\hspre}l<{\hspost}@{}}%
\>[B]{}\mathbf{class}\;\Conid{Profunctor}\;\Varid{p}\Rightarrow \Conid{MonoPro}\;\Varid{p}\;\mathbf{where}{}\<[E]%
\\
\>[B]{}\hsindent{5}{}\<[5]%
\>[5]{}\Varid{mpempty}\mathbin{::}\Varid{p}\;()\;(){}\<[E]%
\\
\>[B]{}\hsindent{5}{}\<[5]%
\>[5]{}(\star)\mathbin{::}\Varid{p}\;\Varid{b}\;\Varid{c}\to \Varid{p}\;\Varid{d}\;\Varid{e}\to \Varid{p}\;(\Varid{b},\Varid{d})\;(\Varid{c},\Varid{e}){}\<[E]%
\ColumnHook
\end{hscode}\resethooks

\noindent satisfying the monoid laws

\begin{itemize}
\item Left  identity: \ensuremath{\Varid{dimap}\;\Varid{diag}\;\Varid{snd}\;(\Varid{mpempty}\star\Varid{f})\mathrel{=}\Varid{f}}
\item Right identity: \ensuremath{\Varid{dimap}\;\Varid{diag}\;\Varid{fst}\;(\Varid{f}\star\Varid{mpempty})\mathrel{=}\Varid{f}}
\item Associativity: \ensuremath{\Varid{dimap}\;assoc^{-1}\;\Varid{assoc}\;(\Varid{f}\star(\Varid{g}\star\Varid{h}))\mathrel{=}(\Varid{f}\star\Varid{g})\star\Varid{h}}
\end{itemize}

\noindent where the helper functions \ensuremath{\Varid{diag}\mathbin{::}\Varid{x}\to (\Varid{x},\Varid{x})},
\ensuremath{assoc^{-1}\mathbin{::}((\Varid{x},\Varid{y}),\Varid{z})\to (\Varid{x},(\Varid{y},\Varid{z}))}, and \ensuremath{\Varid{assoc}\mathbin{::}(\Varid{x},(\Varid{y},\Varid{z}))\to ((\Varid{x},\Varid{y}),\Varid{z})} are the obvious ones.

Another way to understand \ensuremath{\Conid{MonoPro}} is that it lifts pure functions with many
inputs to a binary constructor type, while a profunctor only lifts functions
with one type as input parameter.

\begin{hscode}\SaveRestoreHook
\column{B}{@{}>{\hspre}l<{\hspost}@{}}%
\column{E}{@{}>{\hspre}l<{\hspost}@{}}%
\>[B]{}\Varid{lmap}_{2}\mathbin{::}\Conid{MonoPro}\;\Varid{p}\Rightarrow (\Varid{s}\to (\Varid{a},\Varid{c}))\to \Varid{p}\;\Varid{a}\;\Varid{b}\to \Varid{p}\;\Varid{c}\;\Varid{d}\to \Varid{p}\;\Varid{s}\;(\Varid{b},\Varid{d}){}\<[E]%
\\
\>[B]{}\Varid{lmap}_{2}\;\Varid{f}\;\Varid{pa}\;\Varid{pc}\mathrel{=}\Varid{dimap}\;\Varid{f}\;\Varid{id}\;(\Varid{pa}\star\Varid{pc}){}\<[E]%
\\[\blanklineskip]%
\>[B]{}\Varid{rmap}_{2}\mathbin{::}\Conid{MonoPro}\;\Varid{p}\Rightarrow ((\Varid{b},\Varid{d})\to \Varid{t})\to \Varid{p}\;\Varid{a}\;\Varid{b}\to \Varid{p}\;\Varid{c}\;\Varid{d}\to \Varid{p}\;(\Varid{a},\Varid{c})\;\Varid{t}{}\<[E]%
\\
\>[B]{}\Varid{rmap}_{2}\;\Varid{f}\;\Varid{pa}\;\Varid{pc}\mathrel{=}\Varid{dimap}\;\Varid{id}\;\Varid{f}\;(\Varid{pa}\star\Varid{pc}){}\<[E]%
\ColumnHook
\end{hscode}\resethooks

which can work together as one function 

\begin{hscode}\SaveRestoreHook
\column{B}{@{}>{\hspre}l<{\hspost}@{}}%
\column{E}{@{}>{\hspre}l<{\hspost}@{}}%
\>[B]{}\Varid{rlmap}\mathbin{::}\Conid{MonoPro}\;\Varid{p}\Rightarrow ((\Varid{b},\Varid{d})\to \Varid{t})\to (\Varid{s}\to (\Varid{a},\Varid{c}))\to \Varid{p}\;\Varid{a}\;\Varid{b}\to \Varid{p}\;\Varid{c}\;\Varid{d}\to \Varid{p}\;\Varid{s}\;\Varid{t}{}\<[E]%
\\
\>[B]{}\Varid{rlmap}\;\Varid{f}\;\Varid{g}\;\Varid{pa}\;\Varid{pc}\mathrel{=}\Varid{dimap}\;\Varid{f}\;\Varid{g}\;(\Varid{pa}\star\Varid{pc}){}\<[E]%
\ColumnHook
\end{hscode}\resethooks

\noindent
which is the same behavior as the Day convolution of \ensuremath{\Varid{p}} with itself. Such convolution is the raison d'être of a monoidal profunctor. A parallel composition is followed by a covariant and a contravariant lifting of two pure functions matching the inner structure, which in our case is the product type \ensuremath{(,)}. 

If one chooses a suitable monoidal profunctor \ensuremath{\Varid{p}} and type \ensuremath{\Varid{s}}, 
\ensuremath{\Conid{MonoPro}\;\Varid{p}\;\Varid{s}} inherits the applicative functor behavior naturally. 

\begin{hscode}\SaveRestoreHook
\column{B}{@{}>{\hspre}l<{\hspost}@{}}%
\column{24}{@{}>{\hspre}l<{\hspost}@{}}%
\column{E}{@{}>{\hspre}l<{\hspost}@{}}%
\>[B]{}\Varid{appToMonoPro}\mathbin{::}\Conid{MonoPro}\;\Varid{p}\Rightarrow \Varid{p}\;\Varid{s}\;(\Varid{a}\to \Varid{b})\to \Varid{p}\;\Varid{s}\;\Varid{a}\to \Varid{p}\;\Varid{s}\;\Varid{b}{}\<[E]%
\\
\>[B]{}\Varid{appToMonoPro}\;\Varid{pab}\;\Varid{pa}\mathrel{=}{}\<[24]%
\>[24]{}\Varid{dimap}\;\Varid{diag}\;(\Varid{uncurry}\;(\mathbin{\$}))\;(\Varid{pab}\star\Varid{pa}){}\<[E]%
\ColumnHook
\end{hscode}\resethooks

The \ensuremath{\Conid{MonoPro}} typeclass has a straightforward instance for the Hom profunctor \ensuremath{(\to )} which satisfies the monoidal profunctor laws trivially. A practical use for this instance is writing expressions in a pointfree manner. One can write an \ensuremath{\Varid{unzip'}\mathbin{::}\Conid{Functor}\;\Varid{f}\Rightarrow \Varid{f}\;(\Varid{a},\Varid{b})\to (\Varid{f}\;\Varid{a},\Varid{f}\;\Varid{b})} function, for example, for any functor that has as as input a pair type. A \ensuremath{\Conid{SISO}} is another example of a monoidal profunctor.

\begin{hscode}\SaveRestoreHook
\column{B}{@{}>{\hspre}l<{\hspost}@{}}%
\column{5}{@{}>{\hspre}l<{\hspost}@{}}%
\column{11}{@{}>{\hspre}l<{\hspost}@{}}%
\column{E}{@{}>{\hspre}l<{\hspost}@{}}%
\>[B]{}\mathbf{instance}\;{}\<[11]%
\>[11]{}(\Conid{Functor}\;\Varid{f},\Conid{Applicative}\;\Varid{g})\Rightarrow \Conid{MonoPro}\;(\Conid{SISO}\;\Varid{f}\;\Varid{g})\;\mathbf{where}{}\<[E]%
\\
\>[B]{}\hsindent{5}{}\<[5]%
\>[5]{}\Varid{mpempty}\mathrel{=}\Conid{SISO}\;(\lambda \anonymous \to \Varid{pure}\;()){}\<[E]%
\\
\>[B]{}\hsindent{5}{}\<[5]%
\>[5]{}\Conid{SISO}\;\Varid{f}\star\Conid{SISO}\;\Varid{g}\mathrel{=}\Conid{SISO}\;(\Varid{zip'}\hsdot{\circ }{.}(\Varid{f}\star\Varid{g})\hsdot{\circ }{.}\Varid{unzip'}){}\<[E]%
\ColumnHook
\end{hscode}\resethooks

\noindent where \ensuremath{\Varid{zip'}\mathbin{::}\Conid{Applicative}\;\Varid{f}\Rightarrow (\Varid{f}\;\Varid{a},\Varid{f}\;\Varid{b})\to \Varid{f}\;(\Varid{a},\Varid{b})} is the applicative functor multiplication. The most basic notion of a monoidal profunctor is represented by this instance. It tells us that the input needs to be a functor instance because of \ensuremath{\Varid{unzip'}}, the functions \ensuremath{\Varid{f}} and \ensuremath{\Varid{g}} are composed in a parallel manner using the monoidal profunctor instance for \ensuremath{(\to )} and then regrouped together using the applicative (monoidal) behavior of \ensuremath{\Varid{zip'}}.

\section{Free Monoidal Profunctors}
\label{sec:freemonoidalprofunctor}

There are different theorems with different hypothesis that ensure the existence of a free monoid in a monoidal category.
Here, we follow Rivas and Jaskelioff~\cite{Jask2017} and use the following proposition to ensure the existence of the \emph{free monoidal profunctor}.

\begin{prop}[\cite{Jask2017}]
Let $(\CC, \otimes, I)$ be a monoidal category with internal homs. If $\CC$ has binary coproducts, and for each $A \in ob(\CC)$ the initial algebra for the endofunctor $I + A \otimes -$ exists, then for each A the free monoid $A^{*}$ exists and its carrier is the carrier of the initial algebra.
\end{prop}

The category $Prof(\CC,\CC)$, when $\CC$ is a small monoidal category, is monoidal with the Day convolution $\star$ and the profunctor $I$ as its unit, and also has binary coproducts and internal homs. The least fixed point of the endofunctor $Q(X) = J + P \star X$ in $Prof(\CC,\CC)$ gives the free monoidal profunctor.

Following this definition we arrive at the following implementation of the free monoidal profunctor (see also~\cite{freeMP}). 
\begin{hscode}\SaveRestoreHook
\column{B}{@{}>{\hspre}l<{\hspost}@{}}%
\column{5}{@{}>{\hspre}l<{\hspost}@{}}%
\column{E}{@{}>{\hspre}l<{\hspost}@{}}%
\>[B]{}\mathbf{data}\;\Conid{FreeMP}\;\Varid{p}\;\Varid{s}\;\Varid{t}\;\mathbf{where}{}\<[E]%
\\
\>[B]{}\hsindent{5}{}\<[5]%
\>[5]{}\Conid{MPempty}\mathbin{::}\Varid{t}\to \Conid{FreeMP}\;\Varid{p}\;\Varid{s}\;\Varid{t}{}\<[E]%
\\
\>[B]{}\hsindent{5}{}\<[5]%
\>[5]{}\Conid{FreeMP}\mathbin{::}(\Varid{s}\to (\Varid{x},\Varid{z}))\to ((\Varid{y},\Varid{w})\to \Varid{t})\to \Varid{p}\;\Varid{x}\;\Varid{y}\to \Conid{FreeMP}\;\Varid{p}\;\Varid{z}\;\Varid{w}\to \Conid{FreeMP}\;\Varid{p}\;\Varid{s}\;\Varid{t}{}\<[E]%
\ColumnHook
\end{hscode}\resethooks

\noindent where \ensuremath{\Conid{MPempty}} corresponds to \ensuremath{\Varid{mpempty}}, and \ensuremath{\Conid{FreeMP}} is the multiplication. The multiplication will be apparent if one expands the definition of Day convolution for $P$ and $P^*$. This interface stacks profunctors, and in each layer, it provides pure functions to simulate the parallel composition nature of a monoidal profunctor.

The following functions provide the necessary functions to build the free construction on monoidal profunctors, \ensuremath{\Varid{toFreeMP}} insert a single profunctor into the free structure, and \ensuremath{\Varid{foldFreeMP}} provides a way of evaluating the structure, collapsing it into a single monoidal profunctor.

\begin{hscode}\SaveRestoreHook
\column{B}{@{}>{\hspre}l<{\hspost}@{}}%
\column{E}{@{}>{\hspre}l<{\hspost}@{}}%
\>[B]{}\Varid{toFreeMP}\mathbin{::}\Conid{Profunctor}\;\Varid{p}\Rightarrow \Varid{p}\;\Varid{s}\;\Varid{t}\to \Conid{FreeMP}\;\Varid{p}\;\Varid{s}\;\Varid{t}{}\<[E]%
\\
\>[B]{}\Varid{toFreeMP}\;\Varid{p}\mathrel{=}\Conid{FreeMP}\;\Varid{diag}\;\Varid{fst}\;\Varid{p}\;(\Conid{MPempty}\;()){}\<[E]%
\ColumnHook
\end{hscode}\resethooks

\begin{hscode}\SaveRestoreHook
\column{B}{@{}>{\hspre}l<{\hspost}@{}}%
\column{E}{@{}>{\hspre}l<{\hspost}@{}}%
\>[B]{}\Varid{foldFreeMP}\mathbin{::}(\Conid{Profunctor}\;\Varid{p},\Conid{MonoPro}\;\Varid{q})\Rightarrow (\Varid{p}\leadsto\Varid{q})\to \Conid{FreeMP}\;\Varid{p}\;\Varid{s}\;\Varid{t}\to \Varid{q}\;\Varid{s}\;\Varid{t}{}\<[E]%
\\
\>[B]{}\Varid{foldFreeMP}\;\anonymous \;(\Conid{Arr}\;\Varid{t})\mathrel{=}\Varid{dimap}\;(\mathbin{\char92 \char95 }\to ())\;(\lambda ()\to \Varid{t})\;\Varid{arrr}{}\<[E]%
\\
\>[B]{}\Varid{foldFreeMP}\;(\Conid{Prof}\;\Varid{h})\;(\Conid{FreeMP}\;\Varid{f}\;\Varid{g}\;\Varid{p}\;\Varid{mp})\mathrel{=}\Varid{dimap}\;\Varid{f}\;\Varid{g}\;((\Varid{h}\;\Varid{p})\star\Varid{foldFreeMP}\;(\Conid{Prof}\;\Varid{h})\;\Varid{mp}){}\<[E]%
\ColumnHook
\end{hscode}\resethooks

A free construction behaves like a list and, of course, \ensuremath{\Conid{MonoPro}} should provide a way to embed a plain profunctor into the free context.

\begin{hscode}\SaveRestoreHook
\column{B}{@{}>{\hspre}l<{\hspost}@{}}%
\column{25}{@{}>{\hspre}l<{\hspost}@{}}%
\column{31}{@{}>{\hspre}l<{\hspost}@{}}%
\column{E}{@{}>{\hspre}l<{\hspost}@{}}%
\>[B]{}\Varid{consMP}\mathbin{::}\Conid{Profunctor}\;\Varid{p}{}\<[25]%
\>[25]{}\Rightarrow \Varid{p}\;\Varid{a}\;\Varid{b}\to \Conid{FreeMP}\;\Varid{p}\;\Varid{s}\;\Varid{t}\to \Conid{FreeMP}\;\Varid{p}\;(\Varid{a},\Varid{s})\;(\Varid{b},\Varid{t}){}\<[E]%
\\
\>[B]{}\Varid{consMP}\;\Varid{pab}\;(\Conid{MPempty}\;\Varid{t}){}\<[31]%
\>[31]{}\mathrel{=}\Conid{FreeMP}\;\Varid{id}\;\Varid{id}\;\Varid{pab}\;(\Conid{MPempty}\;\Varid{t}){}\<[E]%
\\
\>[B]{}\Varid{consMP}\;\Varid{pab}\;(\Conid{FreeMP}\;\Varid{f}\;\Varid{g}\;\Varid{p}\;\Varid{fp}){}\<[31]%
\>[31]{}\mathrel{=}\Conid{FreeMP}\;(\Varid{id}\star\Varid{f})\;(\Varid{id}\star\Varid{g})\;\Varid{pab}\;(\Varid{consMP}\;\Varid{p}\;\Varid{fp}){}\<[E]%
\ColumnHook
\end{hscode}\resethooks

\noindent and with it, an instance of \ensuremath{\Conid{MonoPro}} for the free structure can be defined as

\begin{hscode}\SaveRestoreHook
\column{B}{@{}>{\hspre}l<{\hspost}@{}}%
\column{5}{@{}>{\hspre}l<{\hspost}@{}}%
\column{16}{@{}>{\hspre}l<{\hspost}@{}}%
\column{24}{@{}>{\hspre}l<{\hspost}@{}}%
\column{34}{@{}>{\hspre}l<{\hspost}@{}}%
\column{43}{@{}>{\hspre}l<{\hspost}@{}}%
\column{E}{@{}>{\hspre}l<{\hspost}@{}}%
\>[B]{}\mathbf{instance}\;\Conid{Profunctor}\;\Varid{p}\Rightarrow \Conid{MonoPro}\;(\Conid{FreeMP}\;\Varid{p})\;\mathbf{where}{}\<[E]%
\\
\>[B]{}\hsindent{5}{}\<[5]%
\>[5]{}\Varid{mpempty}\mathrel{=}\Conid{MPempty}\;(){}\<[E]%
\\
\>[B]{}\hsindent{5}{}\<[5]%
\>[5]{}\Conid{MPempty}\;\Varid{t}{}\<[16]%
\>[16]{}\star\Varid{q}{}\<[34]%
\>[34]{}\mathrel{=}\Varid{dimap}\;\Varid{snd}\;(\lambda \Varid{x}\to (\Varid{t},\Varid{x}))\;\Varid{q}{}\<[E]%
\\
\>[B]{}\hsindent{5}{}\<[5]%
\>[5]{}\Varid{q}{}\<[16]%
\>[16]{}\star\Conid{MPempty}\;\Varid{t}{}\<[34]%
\>[34]{}\mathrel{=}\Varid{dimap}\;\Varid{fst}\;(\lambda \Varid{x}\to (\Varid{x},\Varid{t}))\;\Varid{q}{}\<[E]%
\\
\>[B]{}\hsindent{5}{}\<[5]%
\>[5]{}(\Conid{FreeMP}\;\Varid{f}\;\Varid{g}\;\Varid{p}\;\Varid{fp}){}\<[24]%
\>[24]{}\star\Varid{fq}{}\<[34]%
\>[34]{}\mathrel{=}\Varid{dimap}\;{}\<[43]%
\>[43]{}(\Varid{assoc}\hsdot{\circ }{.}(\Varid{f}\star\Varid{id}))\;((\Varid{g}\star\Varid{id})\hsdot{\circ }{.}assoc^{-1})\;(\Varid{consMP}\;\Varid{p}\;(\Varid{fp}\star\Varid{fq})){}\<[E]%
\ColumnHook
\end{hscode}\resethooks

\noindent
where \ensuremath{\Varid{assoc}\mathbin{::}((\Varid{x},\Varid{z}),\Varid{c})\to (\Varid{z},(\Varid{x},\Varid{c}))} and \ensuremath{\Varid{associnv'}\mathbin{::}(\Varid{y},(\Varid{w},\Varid{d}))\to ((\Varid{w},\Varid{y}),\Varid{d})}. 

When \ensuremath{\Varid{p}} is an arrow, then \ensuremath{\Conid{FreeMP}\;\Varid{p}} is an arrow. In order to define this instance one needs to collapse all parallel profunctors in order to make the sequential composition. 

\begin{hscode}\SaveRestoreHook
\column{B}{@{}>{\hspre}l<{\hspost}@{}}%
\column{5}{@{}>{\hspre}l<{\hspost}@{}}%
\column{12}{@{}>{\hspre}l<{\hspost}@{}}%
\column{14}{@{}>{\hspre}l<{\hspost}@{}}%
\column{E}{@{}>{\hspre}l<{\hspost}@{}}%
\>[B]{}\mathbf{instance}\;(\Conid{MonoPro}\;\Varid{p},\Conid{Arrow}\;\Varid{p})\Rightarrow \Conid{Category}\;(\Conid{FreeMP}\;\Varid{p})\;\mathbf{where}{}\<[E]%
\\
\>[B]{}\hsindent{5}{}\<[5]%
\>[5]{}\Varid{id}{}\<[14]%
\>[14]{}\mathrel{=}\Conid{FreeMP}\;(\lambda \Varid{x}\to (\Varid{x},()))\;\Varid{fst}\;(\Varid{arr}\;\Varid{id})\;(\Conid{MPempty}\;()){}\<[E]%
\\
\>[B]{}\hsindent{5}{}\<[5]%
\>[5]{}\Varid{mp}\hsdot{\circ }{.}\Varid{mq}{}\<[14]%
\>[14]{}\mathrel{=}\Varid{toFreeMP}\;(\Varid{fromFreeMP}\;\Varid{mp}\Conid{K}.\hsdot{\circ }{.}\Varid{fromFreeMP}\;\Varid{mq}){}\<[E]%
\\[\blanklineskip]%
\>[B]{}\mathbf{instance}\;(\Conid{MonoPro}\;\Varid{p},\Conid{Arrow}\;\Varid{p})\Rightarrow \Conid{Arrow}\;(\Conid{FreeMP}\;\Varid{p})\;\mathbf{where}{}\<[E]%
\\
\>[B]{}\hsindent{5}{}\<[5]%
\>[5]{}\Varid{arr}\;\Varid{f}{}\<[12]%
\>[12]{}\mathrel{=}\Conid{FreeMP}\;(\lambda \Varid{x}\to (\Varid{x},()))\;\Varid{fst}\;(\Varid{arr}\;\Varid{f})\;(\Conid{MPempty}\;()){}\<[E]%
\\
\>[B]{}\hsindent{5}{}\<[5]%
\>[5]{}(\mathbin{***}){}\<[12]%
\>[12]{}\mathrel{=}(\star){}\<[E]%
\ColumnHook
\end{hscode}\resethooks

It is good to remember that the type class \ensuremath{\Conid{Category}} is the one that has the two methods \ensuremath{\Varid{id}} and \ensuremath{(\hsdot{\circ }{.})}, which represents the notion of a category in Haskell.

\section{Monoidal profunctor optics - Monocles}
\label{sec:monocles}


Data accessors are an essential part of functional programming. They allow reading and writing a whole data structure or parts of it \cite{Wu}. In Haskell, one needs to deal with Algebraic Data Types (ADTs) such as products (fields), sums, containers, function types, to name a few. For each of these structures, the action of handling can be a hard task and not compositional at all. To circumvent this problem, different type of data accessors were created, such as lenses, prisms, and traversables~\cite{VL, kmett}. These different notions of modular (composable) data accessors \cite{Wu} were grouped together and called \emph{optics}. Optics help to tackle the data accessor problem with the help of some category-theoretic constructions such as profunctors.

An optic is a general denotation to locate parts (or even the whole) of a data structure in which some action needs to be performed. Each optic deals with a different ADT, for example, the well-known lenses deal with product types, prisms with sum types, traversals with traversable containers, grates with function types, isos deals with any type but cannot change its shape, and so on. A general optic is a polymorphic type on a binary type constructor type with a typeclass restriction \ensuremath{\Varid{r}} on it.

\begin{hscode}\SaveRestoreHook
\column{B}{@{}>{\hspre}l<{\hspost}@{}}%
\column{E}{@{}>{\hspre}l<{\hspost}@{}}%
\>[B]{}\mathbf{type}\;\Conid{Optic}\;\Varid{r}\;\Varid{s}\;\Varid{t}\;\Varid{a}\;\Varid{b}\mathrel{=}\forall \Varid{p}\hsforall \hsdot{\circ }{.}\Varid{r}\;\Varid{p}\Rightarrow \Varid{p}\;\Varid{a}\;\Varid{b}\to \Varid{p}\;\Varid{s}\;\Varid{t}{}\<[E]%
\ColumnHook
\end{hscode}\resethooks

The idea of an optic is to have an in-depth look into get/set operations, for example, if one has a ``big'' data structure \ensuremath{\Varid{s}}, it is possible to extract a piece of it, say \ensuremath{\Varid{a}}, which can be written as a function \ensuremath{\Varid{get}\mathbin{::}\Varid{s}\to \Varid{a}}. Whereas, if one provides a ``big'' structure \ensuremath{\Varid{s}}, and a value \ensuremath{\Varid{b}} (which is a part of \ensuremath{\Varid{s}}), we can modify \ensuremath{\Varid{s}} into another ``big'' structure \ensuremath{\Varid{t}} (the structure might not change and the data still be \ensuremath{\Varid{s}}). This description is modelled by the function \ensuremath{\Varid{set}\mathbin{::}\Varid{s}\to \Varid{b}\to \Varid{t}}.

Both functions (get and set) are specializations of \ensuremath{\Conid{Optic}\;\Conid{Strong}} yielding the type
\begin{hscode}\SaveRestoreHook
\column{B}{@{}>{\hspre}l<{\hspost}@{}}%
\column{E}{@{}>{\hspre}l<{\hspost}@{}}%
\>[B]{}\mathbf{type}\;\Conid{Lens}\;\Varid{s}\;\Varid{t}\;\Varid{a}\;\Varid{b}\mathrel{=}\forall \Varid{p}\hsforall \hsdot{\circ }{.}\Conid{Strong}\;\Varid{p}\Rightarrow \Varid{p}\;\Varid{a}\;\Varid{b}\to \Varid{p}\;\Varid{s}\;\Varid{t}{}\<[E]%
\ColumnHook
\end{hscode}\resethooks

 Here, \ensuremath{\Conid{Strong}} is a profunctor dependent typeclass having the method \ensuremath{\Varid{first'}\mathbin{::}\Varid{p}\;\Varid{a}\;\Varid{b}\to \Varid{p}\;(\Varid{a},\Varid{x})\;(\Varid{b},\Varid{x})} (it also has a method \ensuremath{\Varid{second'}}). For example, one can recover a get function using a lens~\cite{Wu, connor}, and to achieve that we use the profunctor \ensuremath{\Conid{Forget}} that is just the contravariant hom-functor:

\begin{hscode}\SaveRestoreHook
\column{B}{@{}>{\hspre}l<{\hspost}@{}}%
\column{E}{@{}>{\hspre}l<{\hspost}@{}}%
\>[B]{}\mathbf{newtype}\;\Conid{Forget}\;\Varid{r}\;\Varid{a}\;\Varid{b}\mathrel{=}\Conid{Forget}\;\{\mskip1.5mu \Varid{runForget}\mathbin{::}\Varid{a}\to \Varid{r}\mskip1.5mu\}{}\<[E]%
\\[\blanklineskip]%
\>[B]{}\Varid{get}\mathbin{::}\Conid{Lens}\;\Varid{s}\;\Varid{t}\;\Varid{a}\;\Varid{b}\to \Varid{s}\to \Varid{a}{}\<[E]%
\\
\>[B]{}\Varid{get}\;\Varid{lens}\mathrel{=}\Varid{runForget}\;(\Varid{lens}\;(\Conid{Forget}\;\Varid{id})){}\<[E]%
\ColumnHook
\end{hscode}\resethooks

Lenses help to give the intuition behind this profunctorial optics machinery, but this work will solely focus on the optic derived from a monoidal profunctor with
$\otimes = \times$, which combines grates and traversals. It will be called a \emph{monocle}.

\begin{hscode}\SaveRestoreHook
\column{B}{@{}>{\hspre}l<{\hspost}@{}}%
\column{E}{@{}>{\hspre}l<{\hspost}@{}}%
\>[B]{}\mathbf{type}\;\Conid{Monocle}\;\Varid{s}\;\Varid{t}\;\Varid{a}\;\Varid{b}\mathrel{=}\forall \Varid{p}\hsforall \hsdot{\circ }{.}\Conid{MonoPro}\;\Varid{p}\Rightarrow \Varid{p}\;\Varid{a}\;\Varid{b}\to \Varid{p}\;\Varid{s}\;\Varid{t}{}\<[E]%
\ColumnHook
\end{hscode}\resethooks

A \emph{\ensuremath{\Conid{Monocle}}} locates every position of a product (tuple) type (which can be generalized to a finite vector \cite{Jaskelioff2014ART}) while a \ensuremath{\Conid{Lens}} locates parts of it. We now list some basic monocles to observe their behavior. Furthermore, one can observe that those monocles can be generalized if dependent types are used.

\begin{hscode}\SaveRestoreHook
\column{B}{@{}>{\hspre}l<{\hspost}@{}}%
\column{E}{@{}>{\hspre}l<{\hspost}@{}}%
\>[B]{}\Varid{each2}\mathbin{::}\Conid{MonoPro}\;\Varid{p}\Rightarrow \Varid{p}\;\Varid{a}\;\Varid{b}\to \Varid{p}\;(\Varid{a},\Varid{a})\;(\Varid{b},\Varid{b}){}\<[E]%
\\
\>[B]{}\Varid{each2}\;\Varid{p}\mathrel{=}\Varid{p}\star\Varid{p}{}\<[E]%
\\
\>[B]{}\Varid{each3}\mathbin{::}\Conid{MonoPro}\;\Varid{p}\Rightarrow \Varid{p}\;\Varid{a}\;\Varid{b}\to \Varid{p}\;(\Varid{a},\Varid{a},\Varid{a})\;(\Varid{b},\Varid{b},\Varid{b}){}\<[E]%
\\
\>[B]{}\Varid{each3}\;\Varid{p}\mathrel{=}\Varid{dimap}\;\Varid{flat3i}\;\Varid{flat3l}\;(\Varid{p}\star\Varid{p}\star\Varid{p}){}\<[E]%
\\
\>[B]{}\Varid{each4}\mathbin{::}\Conid{MonoPro}\;\Varid{p}\Rightarrow \Varid{p}\;\Varid{a}\;\Varid{b}\to \Varid{p}\;(\Varid{a},\Varid{a},\Varid{a},\Varid{a})\;(\Varid{b},\Varid{b},\Varid{b},\Varid{b}){}\<[E]%
\\
\>[B]{}\Varid{each4}\;\Varid{p}\mathrel{=}\Varid{dimap}\;\Varid{flat4i}\;\Varid{flat4l}\;(\Varid{p}\star\Varid{p}\star\Varid{p}\star\Varid{p}){}\<[E]%
\ColumnHook
\end{hscode}\resethooks

As one can observe, \ensuremath{\Varid{each2}} deals with parallel composition of the argument \ensuremath{\Varid{p}} with itself using the \ensuremath{\Conid{MonoPro}} interface. The focus is on tuples of size 2. The monocles \ensuremath{\Varid{each3}} and \ensuremath{\Varid{each4}} deal with tuples of size 3 and 4 and depend on the tuple flattening functions.

\begin{hscode}\SaveRestoreHook
\column{B}{@{}>{\hspre}l<{\hspost}@{}}%
\column{E}{@{}>{\hspre}l<{\hspost}@{}}%
\>[B]{}\Varid{flat3l}\mathbin{::}((\Varid{a},\Varid{b}),\Varid{c})\to (\Varid{a},\Varid{b},\Varid{c}){}\<[E]%
\\
\>[B]{}\Varid{flat3l}\;((\Varid{a},\Varid{b}),\Varid{c})\mathrel{=}(\Varid{a},\Varid{b},\Varid{c}){}\<[E]%
\\[\blanklineskip]%
\>[B]{}\Varid{flat3i}\mathbin{::}(\Varid{a},\Varid{b},\Varid{c})\to ((\Varid{a},\Varid{b}),\Varid{c}){}\<[E]%
\\
\>[B]{}\Varid{flat3i}\;(\Varid{a},\Varid{b},\Varid{c})\mathrel{=}((\Varid{a},\Varid{b}),\Varid{c}){}\<[E]%
\\[\blanklineskip]%
\>[B]{}\Varid{flat4l}\mathbin{::}(((\Varid{a},\Varid{b}),\Varid{c}),\Varid{d})\to (\Varid{a},\Varid{b},\Varid{c},\Varid{d}){}\<[E]%
\\
\>[B]{}\Varid{flat4l}\;(((\Varid{a},\Varid{b}),\Varid{c}),\Varid{d})\mathrel{=}(\Varid{a},\Varid{b},\Varid{c},\Varid{d}){}\<[E]%
\\[\blanklineskip]%
\>[B]{}\Varid{flat4i}\mathbin{::}(\Varid{a},\Varid{b},\Varid{c},\Varid{d})\to (((\Varid{a},\Varid{b}),\Varid{c}),\Varid{d}){}\<[E]%
\\
\>[B]{}\Varid{flat4i}\;(\Varid{a},\Varid{b},\Varid{c},\Varid{d})\mathrel{=}(((\Varid{a},\Varid{b}),\Varid{c}),\Varid{d}){}\<[E]%
\ColumnHook
\end{hscode}\resethooks

To have a better understanding of those Monocles we adapt the representation theorem of Jaskelioff and O'Connor~\cite{Jaskelioff2014ART} from functors to profunctors. 

\begin{theorem}[\cite{Jaskelioff2014ART}]\label{urt}
(Unary representation for profunctors) Consider an adjunction between profunctors $-^* \vdash U  : \mathcal{E} \to \mathcal{F}$, where $\mathcal{F}$ is small and $\mathcal{E}$ is a full
subcategory of $Prof(Set,Set)$, the family of profunctors $Iso_{A,B}(S,T) = (S \to A) \times (B \to T)$ gives the following isomorphism natural in $A, B$, and dinatural in $S,T$.
\[
\int_{P} UP(A,B) \to UP(S,T) \cong Iso^{*}_{A,B}(S,T),
\]
where $Iso^*$ is the free profunctor generated by $Iso$. Whenever the left-hand side end exists, the isomorphism holds.
\end{theorem}

Since the free monoidal profunctor exists and is of the form \[ P^{*}(S,T) = (J + P \star P^*)(S,T),\] this theorem helps us find the unary representation for monoidal profunctors.

\begin{prop}
The unary representation for monoidal profunctors is given by the isomorphism:
\[
\int_{P} P(A,B) \to P(S,T) \cong \sum\limits_{n \in \mathbb{N}} (S \to A^n) \times (B^n \to T)
\]
where $P$ ranges over all monoidal profunctors.
\begin{proof}[Proof Sketch]
By theorem~\ref{urt}, it is enough to show that \[ U(Iso^{*}_{A,B})(S,T) \cong \sum\limits_{n \in \mathbb{N}} (S \to A^n) \times (B^n \to T).  \]
This is proven using Yoneda, the fact that Day convolution preserves coproducts, and induction on $n$. 

\end{proof}
\end{prop}

The right-hand side of the above isomorphism tells us that the source type can be split as a finite number of copies of a type \ensuremath{\Varid{a}}, and a finite number of copies of a type \ensuremath{\Varid{b}} can be amalgamated in a type \ensuremath{\Varid{t}}, which gives us a semantic for \ensuremath{\Conid{Monocles}} such as \ensuremath{\Varid{each2}}. Using categorical tools to reason about \ensuremath{\Varid{each2}} (the other functions can be understood similarly), one can rewrite it as follows.

\begin{hscode}\SaveRestoreHook
\column{B}{@{}>{\hspre}l<{\hspost}@{}}%
\column{E}{@{}>{\hspre}l<{\hspost}@{}}%
\>[B]{}\Varid{each2}\mathbin{::}\Conid{MonoPro}\;\Varid{p}\Rightarrow (\Varid{s}\to (\Varid{a},\Varid{a}))\to ((\Varid{b},\Varid{b})\to \Varid{t})\to \Varid{p}\;\Varid{a}\;\Varid{b}\to \Varid{p}\;\Varid{s}\;\Varid{t}{}\<[E]%
\\
\>[B]{}\Varid{each2}\;\Varid{f}\;\Varid{g}\;\Varid{p}\mathrel{=}\Varid{dimap}\;\Varid{f}\;\Varid{g}\;(\Varid{p}\star\Varid{p}){}\<[E]%
\ColumnHook
\end{hscode}\resethooks

Since theorem~\ref{urt} holds, one can interchangeably use the \ensuremath{\Conid{Monocle}} or a more easy to reason type as 
\ensuremath{(\Varid{s}\to (\Varid{a},\Varid{a}),(\Varid{b},\Varid{b})\to \Varid{t})}, but the limitation of the latter is obvious.

Actions can now be performed on a monocle, given the desired location; one can read/write any product (tuple) type. 

\begin{hscode}\SaveRestoreHook
\column{B}{@{}>{\hspre}l<{\hspost}@{}}%
\column{E}{@{}>{\hspre}l<{\hspost}@{}}%
\>[B]{}\Varid{foldOf}\mathbin{::}\Conid{Monoid}\;\Varid{a}\Rightarrow \Conid{Monocle}\;\Varid{s}\;\Varid{t}\;\Varid{a}\;\Varid{b}\to \Varid{s}\to \Varid{a}{}\<[E]%
\\
\>[B]{}\Varid{foldOf}\;\Varid{monocle}\mathrel{=}\Varid{runForget}\;(\Varid{monocle}\;(\Conid{Forget}\;\Varid{id})){}\<[E]%
\ColumnHook
\end{hscode}\resethooks

This action tells that given a \ensuremath{\Conid{Monocle}} (location) one can monoidally collect many parts \ensuremath{\Varid{a}} from the big structure \ensuremath{\Varid{s}} (in this case, tuples). For example,
\begin{hscode}\SaveRestoreHook
\column{B}{@{}>{\hspre}l<{\hspost}@{}}%
\column{E}{@{}>{\hspre}l<{\hspost}@{}}%
\>[B]{}\Varid{foldOf}\;\Varid{each3}\mathbin{::}\Conid{Monoid}\;\Varid{a}\Rightarrow (\Varid{a},\Varid{a},\Varid{a})\to \Varid{a}{}\<[E]%
\ColumnHook
\end{hscode}\resethooks
 behaves in the same way as the function \ensuremath{\Varid{fold}} does with lists, its evaluation on the value \ensuremath{(\text{\ttfamily \char34 AA\char34},\text{\ttfamily \char34 BB\char34},\text{\ttfamily \char34 CC\char34})} gives \ensuremath{\text{\ttfamily \char34 AABBCC\char34}} as expected. The list function \ensuremath{\Varid{foldMap}} has a corresponding Monocle called \ensuremath{\Varid{foldMapOf}},

\begin{hscode}\SaveRestoreHook
\column{B}{@{}>{\hspre}l<{\hspost}@{}}%
\column{E}{@{}>{\hspre}l<{\hspost}@{}}%
\>[B]{}\Varid{foldMapOf}\mathbin{::}\Conid{Monoid}\;\Varid{r}\Rightarrow \Conid{Monocle}\;\Varid{s}\;\Varid{t}\;\Varid{a}\;\Varid{b}\to (\Varid{a}\to \Varid{r})\to \Varid{s}\to \Varid{r}{}\<[E]%
\\
\>[B]{}\Varid{foldMapOf}\;\Varid{monocle}\;\Varid{f}\mathrel{=}\Varid{runForget}\;(\Varid{monocle}\;(\Conid{Forget}\;\Varid{f})){}\<[E]%
\ColumnHook
\end{hscode}\resethooks

\noindent which can locate all elements of a 3-element tuple with the expression \begin{hscode}\SaveRestoreHook
\column{B}{@{}>{\hspre}l<{\hspost}@{}}%
\column{E}{@{}>{\hspre}l<{\hspost}@{}}%
\>[B]{}\Varid{foldMapOf}\;\Varid{each3}\mathbin{::}\Conid{Monoid}\;\Varid{r}\Rightarrow (\Varid{a}\to \Varid{r})\to (\Varid{a},\Varid{a},\Varid{a})\to \Varid{r}{}\<[E]%
\ColumnHook
\end{hscode}\resethooks
 as expected.

Every profunctorial optic has a so-called van Laarhoven \cite{OconnorBlg} functorial representation. For a monocle, this representation can be obtained by the following function. 

\begin{hscode}\SaveRestoreHook
\column{B}{@{}>{\hspre}l<{\hspost}@{}}%
\column{E}{@{}>{\hspre}l<{\hspost}@{}}%
\>[B]{}\Varid{convolute}\mathbin{::}(\Conid{Applicative}\;\Varid{g},\Conid{Functor}\;\Varid{f})\Rightarrow \Conid{Monocle}\;\Varid{s}\;\Varid{t}\;\Varid{a}\;\Varid{b}\to (\Varid{f}\;\Varid{a}\to \Varid{g}\;\Varid{b})\to \Varid{f}\;\Varid{s}\to \Varid{g}\;\Varid{t}{}\<[E]%
\\
\>[B]{}\Varid{convolute}\;\Varid{monocle}\;\Varid{f}\mathrel{=}\Varid{unSISO}\;(\Varid{monocle}\;(\Conid{SISO}\;\Varid{f})){}\<[E]%
\ColumnHook
\end{hscode}\resethooks

If we specialize \ensuremath{\Varid{convolute}} using the identity functor \ensuremath{\Varid{f}\mathrel{=}\Conid{Id}}, one gets the definition of a \ensuremath{\Conid{Traversal}}, which is a defined in the lens package \cite{kmett}.

\begin{hscode}\SaveRestoreHook
\column{B}{@{}>{\hspre}l<{\hspost}@{}}%
\column{E}{@{}>{\hspre}l<{\hspost}@{}}%
\>[B]{}\Varid{traverseOf}\mathbin{::}\Conid{Applicative}\;\Varid{g}\Rightarrow \Conid{Monocle}\;\Varid{s}\;\Varid{t}\;\Varid{a}\;\Varid{b}\to (\Conid{Id}\;\Varid{a}\to \Varid{g}\;\Varid{b})\to (\Conid{Id}\;\Varid{s}\to \Varid{g}\;\Varid{t}){}\<[E]%
\\
\>[B]{}\Varid{traverseOf}\;\Varid{monocle}\mathrel{=}\Varid{convolute}\;\Varid{monocle}{}\<[E]%
\ColumnHook
\end{hscode}\resethooks

One can specialize \ensuremath{\Varid{convolute}} using the applicative functor \ensuremath{\Varid{g}\mathrel{=}\Conid{Id}}, to get the van Laarhoven representation for grates (which depends on a Closed type class of Profunctors) \cite{connor}.

\begin{hscode}\SaveRestoreHook
\column{B}{@{}>{\hspre}l<{\hspost}@{}}%
\column{3}{@{}>{\hspre}l<{\hspost}@{}}%
\column{E}{@{}>{\hspre}l<{\hspost}@{}}%
\>[B]{}\mathbf{class}\;\Conid{Profunctor}\;\Varid{p}\Rightarrow \Conid{Closed}\;\Varid{p}\;\mathbf{where}{}\<[E]%
\\
\>[B]{}\hsindent{3}{}\<[3]%
\>[3]{}\Varid{closed}\mathbin{::}\Varid{p}\;\Varid{a}\;\Varid{b}\to \Varid{p}\;(\Varid{x}\to \Varid{a})\;(\Varid{x}\to \Varid{b}){}\<[E]%
\\[\blanklineskip]%
\>[B]{}\Varid{zipFWithOf}\mathbin{::}\Conid{Functor}\;\Varid{f}\Rightarrow \Conid{Monocle}\;\Varid{s}\;\Varid{t}\;\Varid{a}\;\Varid{b}\to (\Varid{f}\;\Varid{a}\to \Conid{Id}\;\Varid{b})\to (\Varid{f}\;\Varid{s}\to \Conid{Id}\;\Varid{t}){}\<[E]%
\\
\>[B]{}\Varid{zipFWithOf}\;\Varid{monocle}\mathrel{=}\Varid{convolute}\;\Varid{monocle}{}\<[E]%
\ColumnHook
\end{hscode}\resethooks
 
Monoidal profunctors with $\otimes = \times$ capture the essence of a grate and a traversal. Grates have a structured contravariant part (input) while traversals, the covariant one (output), while a monocle has both structures.

\section{Effectful Monoidal Profunctors}
\label{sec:effectfulmonoidprofunctor}

The typeclass for monoidal profunctors \ensuremath{\Conid{MonoPro}} is defined in terms of a profunctor \ensuremath{\Varid{p}} over the (fictitious) base category of Haskell types and functions usually known as \emph{Hask}.
However, the Day convolution allows us to use morphisms from other categories, instead of using Hask everywhere. This section presents a generalization of the class \ensuremath{\Conid{MonoPro}} which allows to use morphisms from other categories. We illustrate its use by applying it to morphisms from a Kleisli category, hence allowing effects to be lifted into the structure. The profunctor class will also have a modified form to lift two abstract morphisms instead of pure functions.

\begin{hscode}\SaveRestoreHook
\column{B}{@{}>{\hspre}l<{\hspost}@{}}%
\column{5}{@{}>{\hspre}l<{\hspost}@{}}%
\column{E}{@{}>{\hspre}l<{\hspost}@{}}%
\>[B]{}\mathbf{class}\;\Conid{Category}\;\Varid{k}\Rightarrow \Conid{CatProfunctor}\;\Varid{k}\;\Varid{p}\;\mathbf{where}{}\<[E]%
\\
\>[B]{}\hsindent{5}{}\<[5]%
\>[5]{}\Varid{catdimap}\mathbin{::}\Varid{k}\;\Varid{a}\;\Varid{b}\to \Varid{k}\;\Varid{c}\;\Varid{d}\to \Varid{p}\;\Varid{b}\;\Varid{c}\to \Varid{p}\;\Varid{a}\;\Varid{d}{}\<[E]%
\ColumnHook
\end{hscode}\resethooks

A \emph{\ensuremath{\Conid{CatProfunctor}}} represents a profunctor working with morphisms on an arbitrary category $\CC$ instead of $Hask$, and provides an interface to lift two of those abstract morphisms defined by the binary type constructor \ensuremath{\Varid{k}}. This new class needs to be a multi-parameter type class because of the added constraint \ensuremath{\Conid{Category}}. 

\begin{hscode}\SaveRestoreHook
\column{B}{@{}>{\hspre}l<{\hspost}@{}}%
\column{5}{@{}>{\hspre}l<{\hspost}@{}}%
\column{16}{@{}>{\hspre}l<{\hspost}@{}}%
\column{E}{@{}>{\hspre}l<{\hspost}@{}}%
\>[B]{}\mathbf{class}\;(\Conid{Category}\;\Varid{k},\Conid{CatProfunctor}\;\Varid{k}\;\Varid{p})\Rightarrow \Conid{CatMonoPro}\;\Varid{k}\;\Varid{p}\mid \Varid{p}\to \Varid{k}\;\mathbf{where}{}\<[E]%
\\
\>[B]{}\hsindent{5}{}\<[5]%
\>[5]{}\Varid{cmpunit}{}\<[16]%
\>[16]{}\mathbin{::}\Varid{k}\;\Varid{s}\;()\to \Varid{k}\;()\;\Varid{t}\to \Varid{p}\;\Varid{s}\;\Varid{t}{}\<[E]%
\\
\>[B]{}\hsindent{5}{}\<[5]%
\>[5]{}\Varid{convolute}{}\<[16]%
\>[16]{}\mathbin{::}\Varid{k}\;\Varid{s}\;(\Varid{a},\Varid{c})\to \Varid{k}\;(\Varid{b},\Varid{d})\;\Varid{t}\to \Varid{p}\;\Varid{a}\;\Varid{b}\to \Varid{p}\;\Varid{c}\;\Varid{d}\to \Varid{p}\;\Varid{s}\;\Varid{t}{}\<[E]%
\ColumnHook
\end{hscode}\resethooks

It is good to remember that in the \emph{\ensuremath{\Conid{CatMonoPro}}} class, the type of \ensuremath{\Varid{cmunit}} is isomorphic to \ensuremath{\Varid{p}\;()\;()}, and the type of \ensuremath{\Varid{convolute}} is isomorphic to \ensuremath{\Varid{p}\;(\Varid{a},\Varid{c})\;(\Varid{b},\Varid{d})}. The functional dependency \ensuremath{\Varid{p}\to \Varid{k}} allows to write the unit \ensuremath{\Varid{cmpempty}} having the same role as \ensuremath{\Varid{mpempty}}, and \ensuremath{\mathrel{\star\star}} also having the same role as \ensuremath{\Conid{MonoPro}}'s \ensuremath{\star} satisfying the same laws as seen before. 

\begin{hscode}\SaveRestoreHook
\column{B}{@{}>{\hspre}l<{\hspost}@{}}%
\column{E}{@{}>{\hspre}l<{\hspost}@{}}%
\>[B]{}\Varid{cmpempty}\mathbin{::}\Varid{p}\;()\;(){}\<[E]%
\\
\>[B]{}\Varid{cmpempty}\mathrel{=}\Varid{unitmp}\;\Varid{id}\;\Varid{id}{}\<[E]%
\\[\blanklineskip]%
\>[B]{}(\mathrel{\star\star})\mathbin{::}\Conid{CatMonoPro}\;\Varid{k}\;\Varid{p}\Rightarrow \Varid{p}\;\Varid{a}\;\Varid{b}\to \Varid{p}\;\Varid{c}\;\Varid{d}\to \Varid{p}\;(\Varid{a},\Varid{c})\;(\Varid{b},\Varid{d}){}\<[E]%
\\
\>[B]{}\Varid{p}\mathrel{\star\star}\Varid{q}\mathrel{=}\Varid{convolute}\;\Varid{id}\;\Varid{id}\;\Varid{p}\;\Varid{q}{}\<[E]%
\ColumnHook
\end{hscode}\resethooks

As an example, one can work with \ensuremath{\Conid{CatProfunctor}} and \ensuremath{\Conid{CatMonoPro}} alongside a Kleisli arrow. That is, objects are types but morphisms are Kleisli arrows. The \ensuremath{\Conid{CatProfunctor}} instance in this example will permit computations to be lifted covariantly and contravariantly. The \ensuremath{\Conid{CatMonoPro}} gives a convolutional effect for computations and not just pure functions (as in \ensuremath{\Conid{MonoPro}}'s \ensuremath{\Varid{rlmap}}). 

The Kleisli arrow is just a wrapped type:

\begin{hscode}\SaveRestoreHook
\column{B}{@{}>{\hspre}l<{\hspost}@{}}%
\column{E}{@{}>{\hspre}l<{\hspost}@{}}%
\>[B]{}\mathbf{newtype}\;\Conid{Kleisli}\;\Varid{m}\;\Varid{a}\;\Varid{b}\mathrel{=}\Conid{Kleisli}\;\{\mskip1.5mu \Varid{runKleisli}\mathbin{::}\Varid{a}\to \Varid{m}\;\Varid{b}\mskip1.5mu\}{}\<[E]%
\ColumnHook
\end{hscode}\resethooks

\noindent having a lawful \ensuremath{\Conid{Category}} instance as follows. 

\begin{hscode}\SaveRestoreHook
\column{B}{@{}>{\hspre}l<{\hspost}@{}}%
\column{5}{@{}>{\hspre}l<{\hspost}@{}}%
\column{E}{@{}>{\hspre}l<{\hspost}@{}}%
\>[B]{}\mathbf{instance}\;\Conid{Monad}\;\Varid{m}\Rightarrow \Conid{Category}\;(\Conid{Kleisli}\;\Varid{m})\;\mathbf{where}{}\<[E]%
\\
\>[B]{}\hsindent{5}{}\<[5]%
\>[5]{}\Varid{id}\mathrel{=}\Conid{Kleisli}\;\Varid{return}{}\<[E]%
\\
\>[B]{}\hsindent{5}{}\<[5]%
\>[5]{}(\Conid{Kleisli}\;\Varid{bmc})\hsdot{\circ }{.}(\Conid{Kleisli}\;\Varid{amb})\mathrel{=}\Conid{Kleisli}\;(\lambda \Varid{a}\to (\Varid{amb}\;\Varid{a})\bind \Varid{bmc}){}\<[E]%
\ColumnHook
\end{hscode}\resethooks

A datatype called \ensuremath{\Conid{Lift}} is a \ensuremath{\Conid{CatProfunctor}} with respect to the Kleisli arrow.

\begin{hscode}\SaveRestoreHook
\column{B}{@{}>{\hspre}l<{\hspost}@{}}%
\column{E}{@{}>{\hspre}l<{\hspost}@{}}%
\>[B]{}\mathbf{newtype}\;\Conid{Lift}\;\Varid{t}\;\Varid{m}\;\Varid{a}\;\Varid{b}\mathrel{=}\Conid{Lift}\;\{\mskip1.5mu \Varid{runLift}\mathbin{::}\Varid{m}\;\Varid{a}\to \Varid{t}\;\Varid{m}\;\Varid{b}\mskip1.5mu\}{}\<[E]%
\ColumnHook
\end{hscode}\resethooks
This polymorphic type represents a general version of the function \ensuremath{\Varid{lift}} used to lift monadic computations into a monad transformer~\cite{mtl}. A monad transformer is a way to stack two or more monads together in order to enable more than one effectful computation together~\cite{Jones93composingmonads,mntr}. By packing lift into a profunctor concerning the Kleisli arrow, one gets ways to precompose and post-compose computations with a monad transformer's inner monad. For example, suppose a monad transformer \ensuremath{\Conid{FooT}\;\Varid{t}\;\Varid{m}} (where \ensuremath{\Varid{m}} is a monad). It is possible to use \ensuremath{\Varid{catdimap}} to compose effectful computations in \ensuremath{\Varid{m}} having its results in the monad transformer.

\begin{hscode}\SaveRestoreHook
\column{B}{@{}>{\hspre}l<{\hspost}@{}}%
\column{9}{@{}>{\hspre}l<{\hspost}@{}}%
\column{68}{@{}>{\hspre}l<{\hspost}@{}}%
\column{73}{@{}>{\hspre}l<{\hspost}@{}}%
\column{E}{@{}>{\hspre}l<{\hspost}@{}}%
\>[B]{}\mathbf{instance}\;(\Conid{MonadT}\;\Varid{t},\Conid{Monad}\;\Varid{m},\Conid{Monad}\;(\Varid{t}\;\Varid{m}))\Rightarrow \Conid{CatProfunctor}\;(\Conid{Kleisli}\;\Varid{m})\;(\Conid{Lift}\;\Varid{t}\;\Varid{m})\;\mathbf{where}{}\<[E]%
\\
\>[B]{}\hsindent{9}{}\<[9]%
\>[9]{}\Varid{catdimap}\;(\Conid{Kleisli}\;\Varid{f})\;(\Conid{Kleisli}\;\Varid{g})\;(\Conid{Lift}\;\Varid{h})\mathrel{=}\Conid{Lift}\mathbin{\$}\lambda \Varid{ma}\to {}\<[68]%
\>[68]{}\mathbf{let}\;{}\<[73]%
\>[73]{}\Varid{k}\mathrel{=}\Varid{h}\;(\Varid{ma}\bind \Varid{f}){}\<[E]%
\\
\>[73]{}\Varid{l}\mathrel{=}\Varid{lift}\hsdot{\circ }{.}\Varid{g}{}\<[E]%
\\
\>[68]{}\mathbf{in}\;\Varid{k}\bind \Varid{l}{}\<[E]%
\ColumnHook
\end{hscode}\resethooks

As an important note, \ensuremath{\Conid{Lift}} is a \ensuremath{\Conid{SISO}} with \ensuremath{\Varid{f}\mathrel{=}\Varid{m}}, and \ensuremath{\Varid{g}\mathrel{=}\Varid{t}\;\Varid{m}}. For a plain profunctor, \ensuremath{\Varid{m}} works only with a \ensuremath{\Conid{Functor}}, \ensuremath{\Varid{t}} need not be a monad transformer, and \ensuremath{\Varid{t}\;\Varid{m}} needs only to be an \ensuremath{\Conid{Applicative}}. Hence this instance is substantially different from the mentioned one.

To work with a \ensuremath{\Conid{CatMonoPro}} instance with respect to the Kleisli arrow, a notion of reordering effects (commutativity) is needed. 

\begin{hscode}\SaveRestoreHook
\column{B}{@{}>{\hspre}l<{\hspost}@{}}%
\column{5}{@{}>{\hspre}l<{\hspost}@{}}%
\column{9}{@{}>{\hspre}l<{\hspost}@{}}%
\column{13}{@{}>{\hspre}l<{\hspost}@{}}%
\column{28}{@{}>{\hspre}l<{\hspost}@{}}%
\column{E}{@{}>{\hspre}l<{\hspost}@{}}%
\>[B]{}\mathbf{instance}\;(\Conid{CommT}\;\Varid{t},\Conid{Traversable}\;\Varid{m},\Conid{Monad}\;\Varid{m},\Conid{Monad}\;(\Varid{t}\;\Varid{m}))\Rightarrow {}\<[E]%
\\
\>[B]{}\hsindent{5}{}\<[5]%
\>[5]{}\Conid{CatMonoPro}\;(\Conid{Kleisli}\;\Varid{m})\;(\Conid{Lift}\;\Varid{t}\;\Varid{m})\;\mathbf{where}{}\<[E]%
\\
\>[5]{}\hsindent{4}{}\<[9]%
\>[9]{}\Varid{cmpunit}\;(\Conid{Kleisli}\;\Varid{f})\;(\Conid{Kleisli}\;\Varid{g})\mathrel{=}\Conid{Lift}\;(\lambda \Varid{m}\to \Varid{lift}\;(\Varid{m}\bind \Varid{f}\sequ \Varid{g}\;())){}\<[E]%
\\
\>[5]{}\hsindent{4}{}\<[9]%
\>[9]{}\Varid{convolute}\;(\Conid{Kleisli}\;\Varid{f})\;(\Conid{Kleisli}\;\Varid{g})\;(\Conid{Lift}\;\Varid{h})\;(\Conid{Lift}\;\Varid{l})\mathrel{=}{}\<[E]%
\\
\>[9]{}\hsindent{4}{}\<[13]%
\>[13]{}\Conid{Lift}\mathbin{\$}\lambda \Varid{ms}\to {}\<[28]%
\>[28]{}\mathbf{let}\;(\Varid{ma},\Varid{mc})\mathrel{=}\Varid{unzip'}\;(\Varid{ms}\bind \Varid{f}){}\<[E]%
\\
\>[28]{}\mathbf{in}\;\Varid{comm}\;(\Varid{fmap}\;\Varid{g}\;(\Varid{zip'}\;((\Varid{h}\;\Varid{ma}),(\Varid{l}\;\Varid{mc})))){}\<[E]%
\ColumnHook
\end{hscode}\resethooks

After the \ensuremath{\Varid{fmap}} of function \ensuremath{\Varid{g}}, the remaining expression will have the type \ensuremath{\Varid{t}\;\Varid{m}\;(\Varid{m}\;\Varid{a})}. This can be fixed by building a class giving a function \ensuremath{\Varid{comm}\mathbin{::}(\Conid{Monad}\;\Varid{m},\Conid{Traversable}\;\Varid{m})\Rightarrow \Varid{t}\;\Varid{m}\;(\Varid{m}\;\Varid{a})\to \Varid{t}\;\Varid{m}\;\Varid{a}} to reorder those effects. A traversable instance is used here to provide that swap but other commutativity notions~\cite{Jones93composingmonads} may be used. 

\begin{hscode}\SaveRestoreHook
\column{B}{@{}>{\hspre}l<{\hspost}@{}}%
\column{5}{@{}>{\hspre}l<{\hspost}@{}}%
\column{E}{@{}>{\hspre}l<{\hspost}@{}}%
\>[B]{}\mathbf{class}\;\Conid{MonadT}\;\Varid{t}\Rightarrow \Conid{CommT}\;\Varid{t}\;\mathbf{where}{}\<[E]%
\\
\>[B]{}\hsindent{5}{}\<[5]%
\>[5]{}\Varid{comm}\mathbin{::}(\Conid{Monad}\;\Varid{m},\Conid{Traversable}\;\Varid{m})\Rightarrow \Varid{t}\;\Varid{m}\;(\Varid{m}\;\Varid{a})\to \Varid{t}\;\Varid{m}\;\Varid{a}{}\<[E]%
\ColumnHook
\end{hscode}\resethooks

This setup provides a use of the maybe monad transformer \ensuremath{\Conid{MaybeT}} with 
a monad like \ensuremath{\Conid{Writer}} which is traversable (all necessary type class instances can be found in the \text{\ttfamily mtl}~\cite{mtl}).

\begin{hscode}\SaveRestoreHook
\column{B}{@{}>{\hspre}l<{\hspost}@{}}%
\column{5}{@{}>{\hspre}l<{\hspost}@{}}%
\column{E}{@{}>{\hspre}l<{\hspost}@{}}%
\>[B]{}\mathbf{instance}\;\Conid{CommT}\;\Conid{MaybeT}\;\mathbf{where}{}\<[E]%
\\
\>[B]{}\hsindent{5}{}\<[5]%
\>[5]{}\Varid{comm}\;(\Conid{MaybeT}\;\Varid{mna})\mathrel{=}\Conid{MaybeT}\;(\Varid{mna}\bind \Varid{sequence}){}\<[E]%
\ColumnHook
\end{hscode}\resethooks

Using the monoidal profunctor \ensuremath{\Conid{Lift}} helps us deal with multiple monads at once, enabling us to deal with product types' computations. Consider the effectful function \ensuremath{\Varid{lsplit}}, in writer monad, that splits in two a list of \ensuremath{\Conid{String}} (or any Ord instance) by its order (in this case, lexicographical). One list for elements less or equal to the head, and the other has bigger elements than the head. The effect is logging, telling what the function is doing for debugging purposes.

\begin{hscode}\SaveRestoreHook
\column{B}{@{}>{\hspre}l<{\hspost}@{}}%
\column{5}{@{}>{\hspre}l<{\hspost}@{}}%
\column{76}{@{}>{\hspre}l<{\hspost}@{}}%
\column{E}{@{}>{\hspre}l<{\hspost}@{}}%
\>[B]{}\Varid{lsplit}\mathbin{::}[\mskip1.5mu \Conid{String}\mskip1.5mu]\to \Conid{Writer}\;[\mskip1.5mu \Conid{String}\mskip1.5mu]\;([\mskip1.5mu \Conid{String}\mskip1.5mu],[\mskip1.5mu \Conid{String}\mskip1.5mu]){}\<[E]%
\\
\>[B]{}\Varid{lsplit}\;(\Varid{z}\mathbin{:}\Varid{zs})\mathrel{=}\mathbf{do}{}\<[E]%
\\
\>[B]{}\hsindent{5}{}\<[5]%
\>[5]{}\Varid{xs}\leftarrow \Varid{return}\;(\Varid{filter}\;(\mathbin{<}\Varid{z})\;\Varid{zs}){}\<[E]%
\\
\>[B]{}\hsindent{5}{}\<[5]%
\>[5]{}\Varid{ys}\leftarrow \Varid{return}\;(\Varid{filter}\;(\geq \Varid{z})\;\Varid{zs}){}\<[E]%
\\
\>[B]{}\hsindent{5}{}\<[5]%
\>[5]{}\Varid{tell}\;[\mskip1.5mu \text{\ttfamily \char34 Splitting:~\char34}\plus \Varid{show}\;\Varid{zs}\plus \text{\ttfamily \char34 ~into~\char34}\plus \Varid{show}\;\Varid{xs}\plus \text{\ttfamily \char34 ,~\char34}{}\<[76]%
\>[76]{}\plus \Varid{show}\;\Varid{ys}\mskip1.5mu]{}\<[E]%
\\
\>[B]{}\hsindent{5}{}\<[5]%
\>[5]{}\Varid{return}\;(\Varid{xs},\Varid{ys}){}\<[E]%
\ColumnHook
\end{hscode}\resethooks

The function \ensuremath{\Varid{rsplit}} below just concats two lists and logs this action.

\begin{hscode}\SaveRestoreHook
\column{B}{@{}>{\hspre}l<{\hspost}@{}}%
\column{5}{@{}>{\hspre}l<{\hspost}@{}}%
\column{35}{@{}>{\hspre}l<{\hspost}@{}}%
\column{E}{@{}>{\hspre}l<{\hspost}@{}}%
\>[B]{}\Varid{rsplit}\mathbin{::}\Conid{String}\to ([\mskip1.5mu \Conid{String}\mskip1.5mu],[\mskip1.5mu \Conid{String}\mskip1.5mu])\to \Conid{Writer}\;[\mskip1.5mu \Conid{String}\mskip1.5mu]\;[\mskip1.5mu \Conid{String}\mskip1.5mu]{}\<[E]%
\\
\>[B]{}\Varid{rsplit}\;\Varid{l}\;(\Varid{xs},\Varid{ys})\mathrel{=}\mathbf{do}{}\<[E]%
\\
\>[B]{}\hsindent{5}{}\<[5]%
\>[5]{}\Varid{tell}\;[\mskip1.5mu \text{\ttfamily \char34 Merging:~\char34}\plus \Varid{show}\;\Varid{xs}{}\<[35]%
\>[35]{}\plus \text{\ttfamily \char34 ,~\char34}\plus \Varid{l}\plus \text{\ttfamily \char34 ,~and~\char34}\plus \Varid{show}\;\Varid{ys}\mskip1.5mu]{}\<[E]%
\\
\>[B]{}\hsindent{5}{}\<[5]%
\>[5]{}\Varid{return}\;(\Varid{xs}\plus [\mskip1.5mu \Varid{l}\mskip1.5mu]\plus \Varid{ys}){}\<[E]%
\ColumnHook
\end{hscode}\resethooks

A quicksort can be logged and stopped if an invalid value is found while keeping the logs of what the algorithm did (in this case an empty string). The instance of \ensuremath{\Conid{CatMonoPro}} helps with the splitting, and merging the sorted results inside this \ensuremath{\Conid{MaybeT}}/\ensuremath{\Conid{Writer}} context easily.

\begin{hscode}\SaveRestoreHook
\column{B}{@{}>{\hspre}l<{\hspost}@{}}%
\column{5}{@{}>{\hspre}l<{\hspost}@{}}%
\column{16}{@{}>{\hspre}l<{\hspost}@{}}%
\column{28}{@{}>{\hspre}l<{\hspost}@{}}%
\column{31}{@{}>{\hspre}l<{\hspost}@{}}%
\column{E}{@{}>{\hspre}l<{\hspost}@{}}%
\>[B]{}\Varid{qsort}\mathbin{::}[\mskip1.5mu \Conid{String}\mskip1.5mu]\to \Conid{MaybeT}\;(\Conid{Writer}\;[\mskip1.5mu \Conid{String}\mskip1.5mu])\;[\mskip1.5mu \Conid{String}\mskip1.5mu]{}\<[E]%
\\
\>[B]{}\Varid{qsort}\;[\mskip1.5mu \mskip1.5mu]\mathrel{=}\Varid{return}\;[\mskip1.5mu \mskip1.5mu]{}\<[E]%
\\
\>[B]{}\Varid{qsort}\;\Varid{xs}\mathrel{=}\mathbf{do}{}\<[E]%
\\
\>[B]{}\hsindent{5}{}\<[5]%
\>[5]{}\Varid{guard}\;(\Varid{head}\;\Varid{xs}\not\equiv \text{\ttfamily \char34 \char34}){}\<[E]%
\\
\>[B]{}\hsindent{5}{}\<[5]%
\>[5]{}(\Varid{ls},\Varid{rs}){}\<[16]%
\>[16]{}\leftarrow \Varid{runLift}\;{}\<[28]%
\>[28]{}(\Varid{lconvolute}\;(\Conid{Kleisli}\;\Varid{lsplit})\;\Varid{lift'}\;\Varid{lift'})\;(\Varid{return}\;\Varid{xs}){}\<[E]%
\\
\>[B]{}\hsindent{5}{}\<[5]%
\>[5]{}(\Varid{ls'},\Varid{rs'}){}\<[16]%
\>[16]{}\leftarrow \Varid{runKleisli}\;{}\<[31]%
\>[31]{}((\Conid{Kleisli}\;\Varid{qsort})\star(\Conid{Kleisli}\;\Varid{qsort}))\;(\Varid{ls},\Varid{rs}){}\<[E]%
\\
\>[B]{}\hsindent{5}{}\<[5]%
\>[5]{}\Varid{ss}{}\<[16]%
\>[16]{}\leftarrow \Varid{runLift}\;{}\<[28]%
\>[28]{}(\Varid{rconvolute}\;(\Conid{Kleisli}\;(\Varid{rsplit}\;(\Varid{head}\;\Varid{xs})))\;\Varid{lift'}\;\Varid{lift'})\;(\Varid{return}\;(\Varid{ls'},\Varid{rs'})){}\<[E]%
\\
\>[B]{}\hsindent{5}{}\<[5]%
\>[5]{}\Varid{return}\;\Varid{ss}{}\<[E]%
\ColumnHook
\end{hscode}\resethooks

The functions \ensuremath{\Varid{lconvolute}} and \ensuremath{\Varid{rconvolute}} have \ensuremath{\Varid{id}} (from typeclass \ensuremath{\Conid{Category}}) function in the left and right similar as \ensuremath{\Varid{lmap}_{2}} and \ensuremath{\Varid{rmap}_{2}}. For example, \ensuremath{\Varid{lconvolute}\;\Varid{f}\mathrel{=}\Varid{convolute}\;\Varid{f}\;\Varid{id}} and the right convolution is similar changing the id order. It is also possible to observe the use of a plain monoidal profunctor by using its multiplication since \ensuremath{\Varid{qsort}} is a Kleisli arrow which has a trivial \ensuremath{\Conid{MonoPro}} instance (it is isomorphic to a SISO with \ensuremath{\Varid{f}\mathrel{=}\Conid{Identity}} and \ensuremath{\Varid{g}\mathrel{=}\Conid{MaybeT}\;(\Conid{Writer}\;[\mskip1.5mu \Conid{String}\mskip1.5mu])}). Finally, the expression \ensuremath{\Varid{lift'}} is \ensuremath{\Conid{Lift}\;\Varid{lift}}, that is, the monad transformer's \ensuremath{\Varid{lift}} in the monoidal profunctor setting.

\section{Related Work}
\label{sec:relatedwork}
Several of the results collected in this article are folklore. In the following, we point out the cases where a publication mentioning them is known to us.
Rivas and Jaskelioff introduce notions of computation as monoids~\cite{Jask2017}, and this work follows their path presenting another notion that is not in that work and that has not been well explored in the community. The type class \ensuremath{\Conid{MonoPro}} is defined by Pickering, Gibbons, and Wu \cite{Wu}, but it was not their main focus. The representation of Monocles is also found in the work of Pickering, Gibbons and Wu~\cite{Wu} which states the relationship between a \ensuremath{\Conid{Monocle}} and \ensuremath{\Conid{Traversals}} by using an optic called \ensuremath{\Conid{TraversalP}}. A similar representation of a Monocle was found in the work of O'Connor~\cite{connor} and is called \ensuremath{\Conid{FiniteGrate}}. This representation relies on a type class called \ensuremath{\Conid{Power}}, isomorphic to \ensuremath{\Conid{MonoPro}}, which deals with tuples that use type-level naturals. Also in O'Connor's work, there is a mention of the type \ensuremath{(\Conid{Functor}\;\Varid{f},\Conid{Applicative}\;\Varid{g})\Rightarrow \Varid{f}\;\Varid{a}\to \Varid{g}\;\Varid{b}} which is called \ensuremath{\Conid{SISO}} here. A profunctor optic called \ensuremath{\Conid{Traversables}} is defined by Rom\'an~\cite{roman} which is also similar to Monocles. However, in this work we study for the first time the monoidal profunctor semantics of the mentioned optics. 
The effectful monoidal profunctor and its examples are also new. 

\section{Conclusion}
\label{sec:conclusion}
We have studied monoidal profunctors as monoids in a monoidal category of profunctors, shown that such category is symmetric closed, and obtained the free monoidal profunctor. We have shown how to implement monoidal profunctors in Haskell and have shown some applications related to optics and a generalization to apply monoidal profunctors in a Kleisli category that we call effectful monoidal profunctors. We would like to investigate applications of this extension to parallel and concurrent algorithms.

%

\bibliographystyle{eptcs}
\bibliography{references}

\end{document}